\documentclass[11pt]{article}    
     
\usepackage[bookmarks,colorlinks,breaklinks]{hyperref}  
\hypersetup{linkcolor=blue,citecolor=blue,filecolor=blue,urlcolor=blue} 
\usepackage{fullpage,appendix}
\usepackage{amsmath,amsfonts,amsthm,amssymb,xspace,bm}
\usepackage{dsfont}

\newcommand{\newref}[2][]{\hyperref[#2]{#1~\ref*{#2}}}
\renewcommand{\eqref}[1]{\hyperref[#1]{(\ref*{#1})}}
\numberwithin{equation}{section}

\newcommand{\sref}[1]{\newref[Section]{#1}}
\newcommand{\dref}[1]{\newref[Definition]{#1}}
\newcommand{\tref}[1]{\newref[Theorem]{#1}}
\newcommand{\lref}[1]{\newref[Lemma]{#1}}
\newcommand{\pref}[1]{\newref[Proposition]{prop:#1}}
\newcommand{\cref}[1]{\newref[Corollary]{#1}}
\newcommand{\fref}[1]{\newref[Figure]{#1}}

\newcommand{\eref}[1]{\newref[Equation]{#1}}

\newcommand{\clref}[1]{\newref[Claim]{#1}}

\theoremstyle{plain}
\newtheorem{theorem}{Theorem}[section]
\newtheorem{lemma}[theorem]{Lemma}
\newtheorem{claim}[theorem]{Claim}
\newtheorem{proposition}[theorem]{Proposition}
\newtheorem{corollary}[theorem]{Corollary}

\newtheorem{definition}[theorem]{Definition}

\theoremstyle{definition}

\newtheorem{remark}[theorem]{Remark}
\newtheorem{example}[theorem]{Example}

\newcommand{\bkets}[1]{\left(#1\right)}
\newcommand{\braces}[1]{\left\{#1\right\}}
\newcommand{\sbkets}[1]{\left[#1\right]}

\usepackage{boxedminipage}

\newenvironment{mybox}
{\center \noindent\begin{boxedminipage}{1.0\linewidth}}
{\end{boxedminipage}
\noindent
}
\DeclareMathOperator*{\pr}{\mathsf{Pr}}

\DeclareMathOperator*{\ex}{\mathbb{E}}
\newcommand{\rgta}{\rightarrow}
\newcommand{\lfta}{\leftarrow}
\newcommand{\iprod}[2]{\langle #1,#2\rangle}

\def\inpw#1,#2{\langle #1, #2\rangle}

\newcommand{\reals}{\mathbb{R}}

\newcommand{\R}{\reals}

\newcommand{\note}[1]{\marginpar{\tiny *note in TeX*}}
\newcommand{\ignore}[1]{}

\renewcommand{\phi}{\varphi}
\renewcommand{\epsilon}{\varepsilon}

\newcommand{\zo}{\{0,1\}}

\newcommand{\Mnote}[1]{{}}
\newcommand{\Wnote}[1]{{}}

\newcommand{\pd}{\mathcal{P}(n,2r)}
\newcommand{\pnr}{\mathcal{P}(n,2r)}

\newcommand{\ps}{\mathsf{PS}}
\newcommand{\psd}{\mathsf{PS}(r)}

\newcommand{\bnr}{\binom{[n]}{r}}
\newcommand{\bnlr}{\binom{[n]}{\leq r}}
\newcommand{\sos}{\mathsf{SOS}}
\newcommand{\sa}{\mathsf{SA}}
\newcommand{\ls}{\mathsf{LS}}

\newcommand{\clq}{\mathrm{Clique}}
\newcommand{\m}{\mathcal{M}}
\newcommand{\mm}{M}

\newcommand{\rnr}{\reals^{\binom{[n]}{r} \times \binom{[n]}{r}}}
\newcommand{\rnlr}{\reals^{\binom{[n]}{\leq r} \times \binom{[n]}{\leq r}}}

\date{}
\begin{document}

\title{Sum-of-squares lower bounds for planted clique}
\author{Raghu Meka\\Department of Computer Science\\University of California, Los Angeles\and
Aaron Potechin\thanks{Supported in part by an NSF Graduate Research Fellowship under grant No. 0645960 and by Microsoft Research}\\Massachusetts Institute of Technology\and
Avi Wigderson\thanks{Supported in part by NSF Expeditions grant CCF-0832797}\\Institute for Advanced Study, Princeton}
\maketitle
\begin{abstract}
Finding cliques in random graphs and the closely related ``planted'' clique variant, where a clique of size $k$ is planted in a random $G(n,1/2)$ graph, have been the focus of substantial study in algorithm design. Despite much effort, the best known polynomial-time algorithms only solve the problem for $k = \Theta(\sqrt{n})$. 

In this paper we study the complexity of the {\em planted clique} problem under algorithms from the Sum-Of-Squares hierarchy. We prove the first average case lower bound for this model: for almost all graphs in $G(n,1/2)$, $r$ rounds of the SOS hierarchy cannot find a planted $k$-clique unless $k \geq (\sqrt{n}/\log n)^{1/r}/C^{r}$. Thus, for any constant number of rounds planted cliques of size $n^{o(1)}$ cannot be found by this powerful class of algorithms. This is shown via an integrability gap for the natural formulation of maximum clique problem on random graphs for SOS and Lasserre hierarchies, which in turn follow from degree lower bounds for the Positivestellensatz proof system.

We follow the usual recipe for such proofs. First, we introduce a natural "dual certificate" (also known as a "vector-solution" or "pseudo-expectation") for the given system of polynomial equations representing the problem for every fixed input graph. Then we show that the matrix associated with this dual certificate is PSD (positive semi-definite) with high probability over the choice of the input graph.This requires the use of certain tools. One is the theory of association schemes, and in particular the eigenspaces and eigenvalues of the {\em Johnson scheme}. Another is a combinatorial method we develop to  compute (via traces) norm bounds for certain random matrices whose entries are highly dependent; we hope this method will be useful elsewhere.
\end{abstract}

\section{Introduction}

\subsection{The problem and main result}

Finding cliques in random graphs has been the focus of substantial study in algorithm design. Let $G(n,p)$ denote Erd\"os-Renyi random graphs on $n$ vertices where each edge is kept in the graph with probability $p$. It is easy to check that in a random graph $G \lfta G(n,1/2)$, the largest clique has size $(2+o(1))\log_2 n$ with high probability. On the other hand, the best known polynomial-time algorithms can only find cliques of size $(1+o(1))\log_2 n$ and obtaining better algorithms remains a longstanding open problem: Karp \cite{Karp76} suggested that even finding cliques of size $(1+\epsilon)\log_2 n$ could require superpolynomial time.

Motivated by this, much attention has been given to the related {\sl planted clique problem} or {\sl hidden clique problem} introduced by Jerrum \cite{Jerrum92} and Kucera \cite{Kucera95}. Here, we are given a graph $G \lfta G(n,1/2,k)$ generated by first choosing a $G(n,1/2)$ random graph and placing a clique of size $k$ in the random graph for $t \gg \log_2 n$. The goal is to recover the hidden clique for as small a $k$ as possible given $G$. \Wnote{Added stuff} The study of the planted clique problem and its variations (like finding planted dense subgraphs) is motivated from several other more recent directions. Its potential as being hard on average has lead to proposals to base crypto systems on variants of it~\cite{AppleBW10}. It was used to argue that testing $k$-wise independence is hard near the information theoretic limit by \cite{AlonAKMRX07}. It is used in~\cite{AroraBBG10} to argue that evaluating some financial derivatives is hard. It was also used to justify the hardness of sparse principal component detection by Bethet and Rigollet \cite{BerthetR13}. Another source of interest comes from the related algorithmic problem of finding large communities in social networks. The best known polynomial-time algorithms can solve the problem for $k = \Theta(\sqrt{n})$ \cite{AlonKS98} (see \cite{DGY14} for a near linear-time algorithm) and improving on this bound has received significant attention. The algorithmic problem has also been of much interest in the context of \emph{signal finding} in molecular biology (pattern discovery in DNA sequences) as modeled in the work of \cite{Pevzner00}.

In this work we exhibit a lower bound for the problem in the powerful Lasserre \cite{Lasserre01} and ``sum-of-squares'' ($\sos$) \cite{Parrilo00} semi-definite programming hierarchies\footnote{For brevity, in the following, we will use $\sos$ hierarchy as a common term for the formulations of Lasserre \cite{Lasserre01} and Parrilo \cite{Parrilo00} which are essentially the same in our context.}. As it happens, proving such lower bounds for the planted clique problem reduces easily to proving an integrality gap of value $k$ for the natural formulation of the maximum clique problem in these hierarchies on $G(n,1/2)$ graphs. Our main result then is the following average-case lower bound for maximum clique. \Wnote{Added stuff} We defer the formal definition of the semi-definite relaxation and hierarchies for now, and only note a few facts. First, that implementing the $r$th level of the $\sos$ hierarchy (namely, $r$ rounds), takes roughly $n^{O(r)}$ time, which is polynomial for constant $r$. Second, the above algorithm for $k=\Theta(\sqrt{n})$ may be viewed as implementing only one round. Third, that $r= \log n$ suffices for exact solution of the problem, namely finding the maximum clique.
Our lower bound implies that polynomial time (when the number of rounds $r$ is constant) cannot handle even $k= n^{o(1)}$, and that as many as $(\log n)^{1/2}$ rounds cannot handle $k = (\log n)^{O(1)}$. Here are more precise statements\footnote{Throughout, $c,C$ denote constants.}.    
\begin{theorem}\label{th:mainhierarchy}
With high probability, for $G \lfta G(n,1/2)$ the natural $r$-round $\sos$ relaxation of the maximum clique problem has an integrality gap of at least $n^{1/2r}/C^{r} (\log n)^2$. 
\end{theorem}

As a corollary we obtain the following lower bound for the planted clique problem. 
\begin{corollary}\label{cor:mainhierarchy}
With high probability, for $G \lfta G(n,1/2,t)$ the natural $r$-round $\sos$ relaxation of the planted clique problem has an integrality gap of at least $n^{1/2r}/t C^{r} (\log n)^2$. 
\end{corollary}

\subsection{Background and related work}
Linear and semi-definite hierarchies are one of the most powerful and well-studied techniques in algorithm design. The most prominent of these are the Sherali-Adams hierarchy ($\sa$) \cite{SheraliA90}, Lovasz-Schrijver hierarchy ($\ls$) \cite{LovaszS91}, their semi-definite versions $\sa_+$, $\ls_+$ and Lasserre and $\sos$ hierarchies. The hierarchies present progressively stronger convex relaxations for combinatorial optimization problems parametrized by the number of {\sl rounds} $r$, where the $r$-round relaxation can be solved in $n^{O(r)}$ time on instances of size $n$ in all of them. In terms of relative power (barring some minor technicalities about how the numbering of rounds starts), it is known that $\ls_+(r) < \sa_+(r) < \sos(r)$. Because they capture most powerful techniques for combinatorial optimization, lower bounds for hierarchies serve as strong unconditional evidence for computational hardness. Such lower bounds are even more relevant and compelling in situations where we do not have NP-hardness results, \Wnote{Not sure why the following was removed - I put it back} as is the case for typical average-case optimization problems. 

Broadly speaking, our understanding of the $\sos$ hierarchy is more limited than those of $\ls_+$ and $\sa_+$ hierarchies and in fact the $\sos$ hierarchy appears to be much more powerful. A particularly striking example of this phenomenon was provided by a recent work of Barak et al.~\cite{BarakBHKSZ12}. They showed that a constant number of rounds of the $\sos$ hierarchy can solve the much studied {\sl unique games problem} on instances which need super constant number of $\ls_+,\sa_+$ rounds. It was also shown by the works of \cite{BarakRS11, GuruswamiS11} that the $\sos$ hierarchy captures the sub-exponential algorithm for unique games of \cite{AroraBS10}. These results emphasize the need for a better understanding of the power and limitations of the $\sos$ hierarchy.

From the perspective of proving limitations, all known lower bounds for the $\sos$ hierarchy essentially have their origins in the works of Grigoriev \cite{Grigoriev01b, Grigoriev01}, some of which were later independently rediscovered by Schoenebeck \cite{Schoenebeck08}. These works show that even $\Omega(n)$ rounds of $\sos$ hierarchy cannot solve random $3XOR$ or $3SAT$ instances,  implying a strong unconditional average-case lower bound for a natural distribution. 

Most subsequent lower bounds for $\sos$ hierarchy such as those of \cite{Tulsiani09}, \cite{BhaskaraCVGZ12} rely on \cite{Grigoriev01b} and \cite{Schoenebeck08} and gadget reductions. For example, Tulsiani \cite{Tulsiani09} shows that $2^{O(\sqrt{\log n})}$ rounds of $\sos$ has an integrality gap of $n/2^{O(\sqrt{\log n})}$ for maximum clique in worst-case. This is in stark contrast to the average-case setting: even a single round of $\sos$ gets an integrality gap of at most $O(\sqrt{n})$ for maximum clique on $G(n,1/2)$ \cite{FeigeK00}. Thus, the worst-case and average-case problems have very different complexities. Finally, using reductions tend to induce distributions that are far from uniform and definitely not as natural as $G(n,1/2)$. 

For max-clique on random $G(n,1/2)$ graphs, Feige and Krauthgamer \cite{FeigeK00} showed that $\ls_+(r)$, and hence $\sos(r)$, has an integrality gap of at most $\sqrt{n}/2^{\Omega(r)}$ with high probability. Complementing this, they also showed \cite{FeigeK03} that the gap remains $\sqrt{n}/2^r$ for $\ls_+(r)$ with high probability. However, there were no non-trivial lower bounds known for the stronger $\sos$ hierarchy.

For the planted clique problem, other algorithmic techniques were studied. Jerrum \cite{Jerrum92} showed that a broad class of Markov chain Monte-Carlo (MCMC) based methods cannot solve the problem when the planted clique has size $O(n^{1/2-\delta})$ for any constant $\delta > 0$. Another approach for the planted clique problem based on optimizing a third order tensor was suggested by Frieze and Kannan \cite{FK08}. However, the corresponding optimization problem is NP-hard in the worst-case.

In a recent work, Feldman et al.~\cite{FeldmanGRVX13} introduced the framework of  {\sl statistical algorithms} which generalizes many algorithmic approaches like MCMC methods and showed that such algorithms cannot find large cliques when the planted clique has size $O(n^{1/2-\delta})$ in less than $n^{\Omega(\log n)}$ time\footnote{The results of \cite{FeldmanGRVX13} actually apply to the harder {\sl bipartite} planted clique problem, but this assumption is not too critical.}. However, their framework seems quite different from hierarchy based algorithms. In particular, the statistical algorithms framework is not applicable to algorithms which first pick a sample, fix it, and then perform various operations (such as convex relaxations) on it, as is the case for the hierarchies above. 

Meka and Wigderson \cite{MekaW13} addressed $\sos$ lower bounds for planted clique and claimed a stronger bound than Thm 1.1. While there was a fatal error in their proof, many of the techniques introduced there are used in the present paper. 

Independent of our work, Deshpande and Montanari \cite{DeshpandeM15} recently gave a degree $4$ $\sos$ lower bound for planted clique; while they are only able to handle the degree $4$ case (i.e., $r=2$) , they obtain a better bound for this case than us (roughly $n^{1/3}$ vs $n^{1/4}$ as we do).


\subsection{Proof systems and SDP hierarchies}

A potentially simpler problem than deciding is a large clique exists is the problem of producing short certificates to the non-existence of such cliques. This puts the problem in the realm of proof complexity. Indeed,
we approach the problem of $\sos$ lower bounds from this viewpoint, via the {\sl positivstellensatz} proof system perspective of Grigoriev and Volobjov \cite{GrigorievV01}. We explain this proof system next in general, and then specialize to Boolean problems and specifically to planted clique. 

Suppose we are given a system of polynomial equations or ``axioms'' 
$$f_1(x) = 0,\; f_2(x) = 0,\;\ldots,\;f_m(x) = 0,$$
where each $f_i: \reals^n \rgta \reals$ is a $n$-variate polynomial. A positivstellensatz {\sl refutation} of the system $\mathcal{F} = ((f_i))$ is an identity of the form
$$ \sum_{i=1}^m f_i g_i \equiv 1 + \sum_{i=1}^N h_i^2,$$
where $\{g_1,\ldots,g_m\}$ and $\{h_1,\ldots,h_N\}$ are arbitrary $n$-variate polynomials. Clearly,  if there exists an identity as above, then the system $\mathcal{F}$ has no solution over reals. Starting with the seminal work of Artin on Hilbert's seventeenth problem \cite{Artin27}, a long line of important results in real algebraic geometry -- \cite{Krivine64, Stengle73,Putinar93,Schmudgen91}; cf.~\cite{BCR98} and references therein -- showed that, under some (important) technical conditions\footnote{We avoid going into the details here as the conditions are easily met in the presence of Boolean axioms.}, such certifying identities {\em always exist} for an infeasible system. This motivates the following notion of complexity for refuting systems of polynomial equations.
\begin{definition}[Positivstellensatz Refutation, \cite{GrigorievV01}]\label{dfn:psdrefute}
 Let $\mathcal{F} \equiv \{f_1,\ldots,f_n: \reals^n \rgta \reals\}$, be a system of {\sl axioms}, where each $f_i$ is a real $n$-variate polynomial. A positivstellensatz {\sl refutation} of degree $r$ ($\ps(r)$ refutation, henceforth) for $\mathcal{F}$ is an identity of the form
 \begin{equation}
   \label{eq:psdrefute}
    \sum_{i=1}^m f_i g_i \equiv 1 + \sum_{i=1}^N h_i^2,
 \end{equation}
where $g_1,\ldots,g_m, h_1,\ldots,h_N$ are $n$-variate polynomials such that $deg(f_ig_i) \leq 2r$ for all $i \in [m]$ and $deg(h_j) \leq r$ for all $j \in [N]$. 
\end{definition}

Our interest in positivstellensatz refutations as above comes from the known relations between such identities and $\sos$ hierarchy. Informally (and under appropriate technical conditions), identities as above of degree $r$ show that $\sos$ hierarchy can {\sl certify} infeasibility of the axioms in $2r+\Theta(1)$ rounds and vice versa. We will focus on showing degree lower bounds for identities as above and use them to get integrality gaps for the the $\sos$ hierarchy. We formalize this in \sref{sec:psdtogaps}. For a brief history of the different formulations from \cite{GrigorievV01}, \cite{Lasserre01}, \cite{Parrilo00} and the relations between them and results in real algebraic geometry we refer the reader to \cite{ODonnellZ13}.

Given the above setup, we shall consider the following set of natural axioms to test if a graph $G$ has a clique of size $k$. 

\begin{definition}
  Given a graph $G$, let $\clq(G,k)$ denote the following set of polynomial axioms:
\begin{align}\label{eq:maxclique}
  \text{(Max-Clique): }\;\;\;& x_i^2 - x_i,\;\; \forall i \in [n]\nonumber\\
&x_i \cdot x_j ,\;\; \forall \text{ pairs }\{i,j\} \notin G\\
&\sum_i x_i - k.\nonumber
\end{align}
\end{definition}
Here, the equations on the first line are {\sl Boolean axioms} restricting feasible solutions to be in $\zo^n$. The equations on the second line constrain the support of any feasible $x$ to define a clique in $G$. Finally, the equation on the third line specifies the size of support of $x$. Thus, for any graph $G$, $\clq(G,k)$ is feasible if and only if $G$ has a clique of size $k$. Our core result is to show lower bounds on positivstellensatz refutations for $\clq(G,k)$.

\begin{theorem}[Main]\label{th:main}
With high probability over $G \lfta G(n,1/2)$, the system $\clq(G,k)$ defined by \eref{eq:maxclique} has no $\psd$ refutation for $k \leq n^{1/2r}/C^{r} (\log n)^{1/r}$
\end{theorem}

Given the above theorem it is easy to deduce the integrality gap for the SOS hierarchy, \tref{th:mainhierarchy}: see \sref{sec:psdtogaps}. We next highlight the outline of the proof, and some of our techniques which may be of broader interest. 
\Wnote{Moved techniques to Outline, and made it a subsection of the Intro}

\subsection{Outline}\label{sec:outline}
We now give an outline of our arguments. As in most previous works (cf.~\cite{Grigoriev01}, \cite{Grigoriev01b}, \cite{Schoenebeck08}) on showing lower bounds for $\psd$ refutations, our main tool will be a {\sl dual certificate}. \Wnote{Added}
We note that in the context of hierarchies above, this object is called either a {\em vector solution}\footnote{in which numerical values to variables are replaced by vector values}, or {\em pseudo-expectation}\footnote{reflecting the view of these values as moments of a (possibly nonexistent) probability distribution}. We now turn to define this important notion, which arises naturally from using duality to prove that a degree $r$ refutation like~\ref{eq:psdrefute} does {\em not} exist. 
Let $\pnr:\reals^n \rgta \reals$ be the set of $n$-variate real polynomials of total degree at most $2r$.

\begin{definition}[PSD Mappings]
  A linear mapping $\m:\pnr \rgta \reals$ is said to be positive semi-definite (PSD) if $\m(P^2) \geq 0$ for all $n$-variate polynomials $P$ of degree at most $r$.
\end{definition}

\begin{definition}[Dual Certificates]\label{def:dualcert}
Given a set of axioms $f_1,\ldots,f_m$, a dual certificate for the axioms is a PSD mapping $\m:\pnr \rgta \reals$ such that $\m(f_ig) = 0$ for all $i \in [m]$ and all polynomials $g$ such that $deg(f_i g) \leq 2r$.
\end{definition}

Under reasonable technical conditions which ensure {\sl strong duality}, the converse also holds. For the clique axioms from \eref{eq:maxclique}, a dual certificate would correspond to a feasible vector solution for the $r$-round $\sos$ relaxation for maximum clique (see \fref{fig:lasserresdp} for the exact formulation) with {\sl value} $k$.

The following elementary lemma will be crucial.
\begin{lemma}[Dual Certificate]\label{lm:dualcert}
  Given a system of axioms $((f_i))$, there does not exist a $\psd$ refutation of the system if there exists a dual certificate $\m:\pnr \rgta \reals$ for the axioms.
\end{lemma}

The existence of such a mapping trivially implies a lower bound for $\psd$ refutations: apply $\m$ to both sides of a purported $\psd$ identity as in \eref{eq:psdrefute} to arrive at a contradiction. 

The lemma suggests a general recipe for proving $\psd$ refutation lower bounds: 

\begin{itemize}
\item Design a dual certificate $\m$: For the clique axioms we care about, it is easy to figure out what the {\sl right} dual certificate $\m$ ``should be'' by working backwards from the axioms. The same happens also for the $\psd$ refutation lower bounds of \cite{Grigoriev01, Grigoriev01b}. The main hurdle then is to show that the obtained mapping $\m$ is indeed PSD. At a high level, this reduces to proving a certain random matrix $M \in \reals^{\binom{n}{r} \times \binom{n}{r}}$ is PSD. We show that $M$ is PSD in three steps.
\item Reduction to PSDness of another matrix $M'$: The matrix $M$ has many zero rows and columns which makes it difficult to work with. In \sref{sec:reductiontomprime} we fix this by filling in the zero rows and columns of $M$ to obtain a new matrix $M'$. We then argue that to show $M$ is PSD it is sufficient to show that $M'$ is PSD. 
\item (Deterministic) Matrix analysis: $E = E[M']$ is PSD with a large minimum eigenvalue $\lambda_{min}(E)$. We show this statement in Section \sref{sec:psdexpectation} by using the theory of association schemes described below.
\item Large deviation: with high probability, $\|M'-E\| \leq \lambda_{min}(E)$. This is done by using the structure of our matrix $M'$ along-with a careful application of the trace method to bound the norms of certain random matrices with dependent entries. 
\end{itemize}

We note here the main techniques used.

\paragraph{Techniques: Association schemes}
As discussed, the essence of proving \tref{th:main} involves showing that a certain random matrix is positive semi-definite (PSD) with high probability. In our case, this calls for showing a relation of the form $A \prec B$\footnote{Here and henceforth $\prec$ denotes PSD ordering: $A \prec B$ if and only if $B-A$ is positive definite.} for two matrices $A,B$ whose rows and columns are indexed by subsets of $[n]$ of size $r$. This in turn leads us to matrices which though complicated to describe, will be {\sl set-symmetric} - the entry defined by any two (row and column) sets $I,J$ depends solely on the size of the intersection $I\cap J$. The set of all such matrices, called the {\sl Johnson scheme}, is quite well studied in combinatorics as a special case of {\sl association schemes}. In particular,  all such matrices commute with one another and their common eigenspaces are completely understood. This theory allows us to estimate the eigenvalues and norms of various matrices that arise in the analysis. 
\Mnote{Removed stuff here which does not fit the new proof.}
\ignore{Often, in such situations and especially those involving random matrices, it suffices to show that the smallest eigenvalue of $B$ is bigger than the norm of $A$. This will not be the case for us. Luckily, some of the matrices we study, though complicated to describe, will be {\sl set-symmetric} - the entry defined by any two (row and column) sets $I,J$ depends solely on the size of the intersection $I\cap J$. The set of all such matrices, called the {\sl Johnson scheme}, is quite well studied in combinatorics as a special case of {\sl association schemes}. In particular,  all such matrices commute with one another and their common eigenspaces are completely understood. This theory allows us to show that the eigenvalue of $B-A$ is non-negative on each of the specific eigenspaces of the Johnson scheme.}

\Wnote{Added}
\paragraph{Techniques: Trace bounds for locally random matrices} After various simplifications and reductions, a central problem we have to deal with is upper bounding the spectral norm of certain random matrices, defined by the underlying random graph $G \lfta G(n,1/2)$. As above, these matrices have rows and columns indexed by subsets of vertices. The entry $(I,J)$ of the matrix will be a random variable of expectation zero, which depends only on the edges and non-edges of $G$ in the subgraph induced by $I\cup J$ (hence we name such matrices {\em local}). In the simple case when $r=1$ (so rows and columns are indexed by singletons), which is the one studied in the analysis of the $\sqrt{n}$ approximation algorithm, the random variables in all entries are mutually independent, and a norm bound is easy to obtain by a straightforward use of the trace method. However, for $r>1$ as we need to handle, the entries of the matrix are dependent whenever the edge sets of their entries intersect. This significantly complicates the trace calculation, and we develop some combinatorial tools to bound the trace of high powers of such local matrices.

\Mnote{Some details here.}

\Mnote{Write more details for the large deviation part.}

\section{Dual certificate for $\psd$ refutations of max-clique}
We will specify the dual certificate $\m$ by defining it for polynomials where each individual variable has degree at most $1$ and {\sl extend} $\m$ multi-linearly to all polynomials: for any polynomial $P$, $\m(P) = \m(\tilde{P})$ where $\tilde{P}$ is obtained from $P$ by reducing the individual degrees of all variables to $1$. We can do this without loss of generality because of the Boolean axioms. 

As mentioned in the introduction, we can often work out what the dual certificate should be from the axioms and basic linear algebra. As an example, we first work out the case where the graph $G$ is the complete graph; this will also help us draw a concrete connection to the work of \cite{Grigoriev01}. 

\subsection{Complete graph and knapsack}
For complete graph, the clique axioms simplify to 
\begin{align}\label{eq:maxclique}
  \text{(Max-Clique): }\;\;\;& x_i^2 - x_i,\;\; \forall i \in [n]\nonumber\\
&\sum_i x_i - k.\nonumber
\end{align}
These incidentally also correspond to proving lower bounds for knapsack as studied by Grigoriev \cite{Grigoriev01} (and was what lead us to the specific dual certificate we study). However, in the context of lower bounds for knapsack, the axioms are mainly interesting for non-integer $k$ and Grigoriev shows that for  non-integer $k \leq n/2$, the above system has no $\psd$ refutation for $r < k$. 

\newcommand{\mgr}{\m_{Gr}}

The above axioms tell us that any candidate dual certificate $\mgr \equiv :\pnr \rgta \reals$ should satisfy:
\begin{align*}
  &\mgr\left(\left(\sum_{i=1}^n x_i -k\right)\left(\prod_{i \in I} x_i\right)\right) = 0, \text{ $\forall I,\,|I| < 2r$.}
\end{align*}
For $I \subseteq [n]$, let $X_I = \prod_{i \in I} x_i$. Now, as the above equation is symmetric, it is natural to assume that $\mgr$ is also symmetric in the sense that $\mgr(X_I) = f(|I|)$ for some function $f:\{0,\ldots,2r\} \to \reals_+$. Working from this assumption, Grigoriev derives the following recurrence relation for $f:\{0,\ldots,2r\} \to \reals_+$,
$$f(i+1) = \frac{k-i}{n-i} f(i).$$
From the above it follows that we can define $f$ and hence $\m$ as follows:
$$\mgr(X_I) = f(|I|) = f(0) \cdot \frac{k(k-1)\cdots (k-|I|)}{n(n-1)\cdots (n-|I|)}$$
Grigoriev takes $f(0) = 1$. Here we set $f(0) = \binom{n}{2r}$ with a view towards what is to come. Thus, the final certificate is
\begin{equation}
  \label{eq:grigoriev}
\mgr(X_I) =  \binom{n}{2r} \cdot \frac{k(k-1)\cdots (k-|I|)}{n(n-1)\cdots (n-|I|)} = \binom{n-|I|}{2r-|I|} \cdot \frac{\binom{k}{|I|}}{\binom{2r}{|I|}}.  
\end{equation}
Grigoriev shows the following:
\begin{theorem}[\cite{Grigoriev01}]\label{th:grigknapsack}
For $k < n/2$, the mapping $\mgr$ defined above is PSD for $r < k$. 
\end{theorem}
\subsection{Certificate for clique axioms}
Following a similar approach, we now derive the dual certificate for the clique axioms from Equations \ref{eq:maxclique}, which we restate below for convenience: given a graph $G$ on $n$ vertices, $k \leq n$, the axioms of $\clq(G,k)$ are
\begin{align}
 \text{(Max-Clique): }\;\;\;& x_i^2 - x_i,\;\; \forall i \in [n]\nonumber\\
&x_i \cdot x_j ,\;\; \forall \text{ pairs }\{i,j\} \notin G\\
&\sum_i x_i - k.\nonumber
\end{align}

The above axioms tell us that any candidate dual certificate $\m \equiv \m_G:\pnr \rgta \reals$ should satisfy:
\begin{align}\label{eq:dualcert}
&\m\left(X_I\right) = 0, \text{ $\forall I,\,|I| \leq 2r$, $I$ is not a clique in $G$},\nonumber\\
&\m\left(\left(\sum_{i=1}^n x_i -k\right)X_I \right) = 0, \text{ $\forall I,\,|I| < 2r$.}
\end{align}

The above equations give us a system of linear equations that $\m$ needs to satisfy. By working with the equations, it is easy to {\sl guess} a natural solution for the system. 

Given a graph $G$ on $[n]$, and $I \subseteq [n]$, $|I| \leq 2r$, let 
$$deg_G(I) = |\{S \subseteq [n]: I \subseteq S,\; |S| = 2r, \text{ $S$ is a clique in $G$}\}|.$$
For instance, if $r=1$ and $v \in G$, then $deg_G(\{v\})$ is the degree of vertex $v$. 

We define $\m\equiv \m_G:\pd \rgta \reals$ for monomials as follows: for $I \subseteq [n], |I| \leq 2r$, let
\begin{equation}
  \label{eq:dualsol}
\m\left(\prod_{i \in I} x_i\right) = deg_G(I) \cdot \frac{k (k-1) \cdots (k- |I | + 1)}{2r(2r-1)\cdots (2r-|I| + 1)} = \deg_G(I) \cdot \frac{\binom{k}{|I|}}{\binom{2r}{|I|}}.
\end{equation}
It is easy to check the following claim:
\begin{claim}\label{clm:dualsol}
For any graph $G$, $\m \equiv \m_G$ defined by \eref{eq:dualsol} satisfies Equations \ref{eq:dualcert}.
\end{claim}
\begin{proof}
The first equation in \eref{eq:dualcert} follows immediately from the definition of $\m$. Now, for  $I \subseteq [n], |I| < 2r$,
  \begin{multline*}
    \m\left(\left(\sum_i x_i - k\right)X(I)\right) = (|I| - k) \m(X(I)) + \sum_{j \notin I} \m(X(I \cup \{j\}))\\
= (|I| - k) \cdot deg_G(I) \cdot \frac{\binom{k}{|I|}}{\binom{2r}{|I|}} + \sum_{j \notin I} deg_G(I \cup \{j\}) \cdot \frac{\binom{k}{|I|+1}}{\binom{2r}{|I|+1}}\\
= \frac{\binom{k}{|I|+1}}{\binom{2r}{|I|+1}} \cdot \left(- (2r-|I|) \cdot deg_G(I) + \sum_{j \notin I} deg_G(I \cup \{j\})\right).
  \end{multline*}
Observe that our notion of {\it degree}, $deg_G$, satisfies the following recurrence: for $|I| < 2r$,
\[ deg_G(I) = \frac{1}{2r-|I|} \cdot \sum_{j \notin I, \text{ j adjacent to all of $I$}} deg_G(I \cup \{j\}) = \frac{1}{2r-|I|} \sum_{j \notin I} deg_G(I \cup \{j\}).\]
The above two equations imply that $\m$ satisfies the second equation in \ref{eq:dualcert}.
\end{proof}

Thus, to prove our main theorem \tref{th:main}, it suffices to show that $\m$ as defined above is PSD with high probability. 
We now argue that in fact, to show that $\m$ is PSD we do not need to consider all polynomials $P$ of degree at most $r$. Rather, it is sufficient to show that $\m(P_1^2) \geq 0$ whenever $P_1$ is multilinear and homogeneous of degree $r$. 
\begin{lemma}\label{lem:homogeneousreduction}
For any $P$ of degree at most $r$ we may write $P = P_1 + \sum_{i}{P_{2i}(x^2_i - x_i)} + P_{3}(\sum_{i}{x_i} - k)$ where $P_1$ is multilinear and homogeneous of degree $r$, $P_3$ has degree at most $r-1$, and all $P_{2i}$ have degree at most $r-2$.
\end{lemma}
\begin{proof}
We first make $P$ multilinear by removing any terms which are not multilinear from $P$ as follows. If $P$ has a term of the form ${x^2_i}f$ where $f$ has degree at most $r-2$, write ${x^2_i}f = (x^2_i - x_i)f + {x_i}f$. Iteratively applying this procedure, we may write $P = P'$ plus terms of the form $(x^2_i - x_i)f$ where $P'$ is multilinear of degree at most $r$ and 
$f$ has degree at most $r-2$.

We now make $P'$ multilinear and homogeneous of degree $r$ by removing any terms which have lower degree as follows. If $P'$ has a term of the form $X_I$ where $|I| < r$, write 
$$X_I = \frac{1}{|I| - k}\left(\sum_{i}{x_{i} - k}\right)X_I + \frac{1}{k - |I|}\sum_{i \in I}{(x^2_i - x_i)X_{I \setminus \{i\}}} + \frac{1}{k - |I|}\sum_{i \notin I}{X_{I \cup \{i\}}}$$
Iteratively applying this procedure, we may write $P = P_1$ plus terms of the form $(x^2_i - x_i)f$ and terms of the forms 
$(\sum_{i}{x_i - k})g$ where $P_1$ is multilinear and homogeneous of degree $r$, all such $f$ have degree at most $r-2$ and 
all such $g$ have degree at most $r-1$. Putting everything together, the result follows.
\end{proof}
\begin{corollary}
If $\m(P_1^2) \geq 0$ for all multilinear homogeneous $P_1$ of degree $r$ then $\m$ is PSD.
\end{corollary}
\begin{proof}
Assume $\m(P_1^2) \geq 0$ for all multilinear homogeneous $P_1$ of degree $r$ and $\m(P^2) < 0$ for some $P \in \mathcal{P}(n,r)$. Using \lref{lem:homogeneousreduction}, we may write 
$P = P_1 + \sum_{i}{P_{2i}(x^2_i - x_i)} + P_{3}(\sum_{i}{x_i} - k)$ where $P_1$ is multilinear and homogeneous of degree $r$. $\m(P^2) = \m(P_1^2)$ so $\m(P_1^2) < 0$. Contradiction.
\end{proof}
Thus, showing that $\m$ is PSD with high probability is equivalent to showing that the following matrix $\mm \equiv \mm_G \in \rnr$ is PSD with high probability for $G \lfta G(n,1/2)$: for $I,J \in \binom{[n]}{r}$,
\begin{equation}
  \label{eq:mainmatrixdef}
  \mm(I,J) = deg_G(I \cup J) \cdot \frac{\binom{k}{|I \cup J|}}{\binom{2r}{|I \cup J|}}.
\end{equation}

In the remainder of the paper, we show that $\mm$ is PSD with high probability for $k \leq \Omega_r(n^{1/{2r}}/(\log n)^{1/r})$.
\begin{theorem}[Main Technical Theorem]\label{th:mainpsd}
There exists a constant $c > 0$ such that, with high probability over $G \lfta G(n,1/2)$, the matrix $\mm_G$ defined by \eref{eq:mainmatrixdef} is PSD for $k \leq 2^{-cr} \cdot (\sqrt{n}/\log n)^{1/r}$.
\end{theorem}

\section{Overview of proof of \tref{th:mainpsd}}\label{sec:blurb}
The proof of \tref{th:mainpsd} is quite technical, and is broken into two parts, where the second part is further broken down into smaller parts. While we gave a sketch of the proof of \tref{th:mainpsd} in the inroduction, we give a more detailed overview of the proof here. Recall that all matrices mentioned below are random matrices which are specified by the choice of the random graph $G$.

As mentioned in the introduction, the matrix $M=M_G$ has many zero rows and columns which makes it difficult to work with. The first part is to fill in the zero rows and columns of $M$ to obtain a new matrix, $M'$, which is nonsingular and has no high variance entries. In \sref{sec:reductiontomprime} we define this matrix $M'$ and show that if $M'$ is PSD, so is $M$. The idea is that $M$ and $M'$ are symmetric and the nonzero part of $M$ is a principal submatrix of $M'$, so the smallest nonzero eigenvalue of $M$ is at least as large as the smallest eigenvalue of $M'$.

The second part is to prove that $M'$ is PSD (indeed we prove that it has a high positive smallest eigenvalue). This is stated in the main technical lemma \lref{lm:maintech}.
For the proof of \lref{lm:maintech} we decompose the matrix $M'$ as $M' = E + L + \Delta$, where (a) $E = \ex[M']$ is the expectation matrix; (b) $L$ will be a ``local'' random matrix such that for sets $I,J$, $L(I,J)$ only depends on the edges among the vertices of $I \cup J$ and (c) $\Delta$ is a ``global'' error matrix whose entries are small in magnitude.

Having defined $E$ (which is set-symmetric), let us spell out what the other matrices are. The``local" random matrix $L$ is defined in a simple way as follows:
\begin{align*}
  L(I,J) =
  \begin{cases}
    - E(I,J) &\text{ if some edge in $\mathcal{E}(I \cup J) \setminus (\mathcal{E}(I) \cup \mathcal{E}(J))$ is missing from $G$}\\
 \beta(|I \cap J|) &\text{ otherwise}
  \end{cases},
\end{align*}
where $\mathcal{E}(I)$ denotes the set of possible edges between vertices of $I$ and $\beta:\{0,\ldots,r\} \to \R_+$ is suitably chosen so that each individual entry of $L$ has expectation zero.

\ignore{
If $I,J$ are not disjoint, than $L(I,J)=0$. Otherwise, $L(I,J)= -1$, if some edge in $\mathcal{E}(I \cup J) \setminus (\mathcal{E}(I) \cup \mathcal{E}(J))$ is missing from the graph $G$. If all these edges are present in $G$ then $L(I,J) = 2^{r^2}-1.$
Observe that each individual entry of $L$ has expectation zero.}

Finally, define the last matrix $\Delta = M' - E - L$. 

The proof that $M'$ is PSD proceeds in three modular steps:
\begin{enumerate}
\item We use the results about Johnson scheme to show that $E \succ 0$ and has a large least eigenvalue (roughly $\Omega_r(k^r n^r)$); see \sref{sec:psdexpectation}.
\item We next show that $\|L\| < Ck^{2r} n^{r-1/2}\log n$ by exploiting the recursive structure of the matrix $L$ and some careful trace calculations. This is the most technically intensive part of the proof, and requires the development of some combinatorial tools to estimate the trace of high powers of $L$; see \sref{sec:normboundL}.
\item We then show that $\|\Delta\| < Ck^{2r} n^{r-1/2}\log n$. This is done by first showing that {\em every} entry of $\Delta$ is small in magnitude, via concentration bounds on the number of cliques in random graphs, and bounding its norm using Gershgorin's circle theorem (\lref{lm:Gershgorin}); see \sref{sec:normboundDelta}.
\end{enumerate}



\section{Preliminaries}\label{sec:prelims}
We shall use the following notations\footnote{Some are repeated from the introduction so as to have them at one place.}:
\begin{enumerate}
\item $\pd$ denotes the set of $n$-variate polynomials of degree at most $2r$.
\item $\psd$ denotes positivstellensatz refutations of degree at most $r$ as defined in \dref{dfn:psdrefute}.
\item A linear mapping $\m:\pnr \rgta \reals$ is said to be positive semi-definite (PSD) if $\m(P^2) \geq 0$ for all $P \in \mathcal{P}(n,r)$.
\item For $0 \leq r \leq n$, let $\binom{[n]}{r}$, $\binom{[n]}{\leq r}$ denote all subsets of size exactly and at most $r$, respectively.
\item For $0 \leq r \leq n$, $\reals^{\binom{[n]}{r} \times \binom{[n]}{r}}$ denotes matrices with rows and columns indexed by subsets of $[n]$ of size exactly $r$. Similarly, $\reals^{\binom{[n]}{\leq r} \times \binom{[n]}{\leq r}}$ denotes matrices with rows and columns indexed by subsets of $[n]$ of size at most $r$.   
\item We will view linear functionals $\m:\pnr \rgta \reals$ as matrices $\mm \in \rnlr$, where for $I,J \in \binom{[n]}{\leq r}$, $\mm_{IJ} = \m\left (\prod_{s \in I \cup J} x_s\right)$. In general, this correspondence is not bijective. However, as we only deal with mappings which are constant under multi-linear extensions throughout, the correspondence is one-to-one. It is a standard fact that a mapping $\m$ is PSD if and only if the matrix $\mm$ is PSD. 
\item For $I \subseteq [n]$, let $X_I = \prod_{i \in I} x_i$.
\item By default all vectors are column vectors. For a set $I$, $\mathds{1}(I)$ denotes the indicator vector of the set $I$.
\item For a matrix $A \in \reals^{m \times n}$, $A^\dagger \in \reals^{n \times m}$ denotes its conjugate matrix.    
\end{enumerate}
We will also need the following standard fact from matrix theory (see \cite{Golub1996} for instance).
\begin{lemma}[special case of Gershgorin circle theorem]\label{lm:Gershgorin}
For any square matrx $M \in \R^{N \times N}$, 
$$\|M\| \leq \max_{i \in [N]} \left(\sum_{j=1}^N |M_{ij}|\right).$$
\end{lemma}

Finally, we need McDiarmid's inequality for obtaining tail bounds for functions of independent random variables (see \cite{dubashi2009concentration} for instance)
\begin{theorem}[McDiarmid's inequality]\label{th:mcdiarmid}
Let $X_1,\ldots,X_n$ be independent random variables and let $f$ be a function over the domain space of $(X_1,\ldots,X_n)$. Let $c_1,\ldots,c_n > 0$ be such that for all $i$, $x_1,\ldots,x_n, x_i'$, 
$$|f(x_1,\ldots,x_{i-1},x_i,x_{i+1},\ldots,x_n) - f(x_1,\ldots,x_{i-1},x_i',x_{i+1},\ldots,x_n)| \leq c_i.$$
Then, for all $t > 0$, 
$$\pr\sbkets{|f(X_1,\ldots,X_n) - \ex[f]| > t} \leq 2 \exp\bkets{\frac{-2t^2}{\sum_{i=1}^n c_i^2}}.$$
\end{theorem}
\section{Reduction to PSDness of $M'$}\label{sec:reductiontomprime}
In this section, we define the matrix $M'$ and show that if $M'$ is PSD then so is $M$. We use the following notations for brevity: For any set $I \subseteq [n]$, let $\mathcal{E}(I) = \{\{i,j\}: i\neq j \in I\}$. For $0 \leq i \leq r$, let 
\begin{equation}
\beta(i) = \binom{k}{2r-i}/\binom{2r}{2r-i}.
\end{equation}

For every $T \subseteq [n]$, let $M_T \in \reals^{\binom{[n]}{r} \times \binom{[n]}{r}}$, with $M_T(I,J) = \beta(|I \cap J|)$ if $I \cup J \subseteq T$, and $G$ contains every edge in $\mathcal{E}(T) \setminus \mathcal{E}(I) \cup \mathcal{E}(J)$ (i.e., the only edges in $T$ missing in $G$ are those with both end points in one of $I$ or $J$). We will study the matrix 
\begin{equation}\label{eq:mr}  
M' = \sum_{T: |T| = 2r} M_T.
\end{equation}

Intuitively, for every $I,J$, $M'(I,J)$  is what $M(I,J)$ would be had we added cliques on the subsets $I$, $J$ to the graph. The above definition avoids the problem of the whole row and column corresponding to $I$ or $J$ becoming zero if either was not a clique and controls the variance of the entries. We now show that to show $M$ is PSD, it is sufficient to show that $M'$ is PSD.
\begin{lemma}\label{lm:MprimePSDtoMPSD}
If $M'$ is PSD then $M$ is PSD.
\end{lemma}
\begin{proof}
The reason this lemma is true is because as shown below, the nonzero part of $M$ is a principal submatrix of $M'$.
\begin{proposition}
Whenever $I$ and $J$ are cliques of size $r$ in $G$, $M'(I,J) = M(I,J)$
\end{proposition}
\begin{proof}
Suppose that $I$ and $J$ are cliques in $G$. Then, $M_T(I,J) = \beta(|I \cap J|)$ if $I \cup J \subseteq T$ and $T$ is a clique and $0$ otherwise. Therefore, 
$$M'(I,J) = \sum_{T} M_T(I,J) = \beta(|I \cap J|)\cdot \left|\{T: I \cup J \subseteq T, T \text{ clique}\}\right| = M(I,J).$$
\end{proof}
\begin{corollary}\label{cor:MandMprime}
The nonzero part of $M$ is a principal submatrix of $M'$.
\end{corollary}
We now use the following elementary fact about matrices.
\begin{proposition}\label{prop:principalsubmatrix}
If $A$ is a principal submatrix of a symmetric matrix $B$ then the smallest eigenvalue of $A$ is at least as large as the smallest eigenvalue of $B$. 
\end{proposition}
\begin{proof}
Without loss of generality, $A$ is an $l \times l$ matrix and $B$ is an $m \times m$ matrix where $l \leq m$. Let $v \in \mathbb{R}^{l}$ be a unit eigenvector of $A$ with minimal eigenvalue $\lambda_{min}$. If we let $w \in \mathbb{R}^m$ be the extension of $v$ to $\mathbb{R}^m$ with zeros in the other coordinates, ${w^T}Bw = {v^T}Av = \lambda_{min}$. This implies that the smallest eigenvalue of $B$ is at most $\lambda_{min}$ and the result follows. 
\end{proof}
Combining \cref{cor:MandMprime} and \pref{principalsubmatrix}, if $M'$ is PSD then $M$ is PSD, as needed.
\end{proof}
\newcommand{\jh}{\mathcal{J}}
\section{Johnson  scheme}\label{sec:jscheme}
Association schemes is a classical area in combinatorics and coding theory (cf.~for instance \cite{VanLintW01}). We shall use a few classical results (lemmas \ref{lm:eigenschemes}, \ref{lm:qupperbound} below), about the eigenspaces and eigenvalues of association schemes and the {\sl Johnson scheme} in particular. We also introduce two bases for the Johnson scheme, which will play a key role in bounding the eigenvalues of various matrices later.

We start with some basics about the Johnson scheme - some of our notations are  non-standard but they fit better with the rest of the manuscript.

\begin{definition}[Set-Symmetry]
A matrix $M \in \rnr$ is set-symmetric if for every $I,J \in \binom{[n]}{r}$, $M(I,J)$ depends only on the size of $|I \cap J|$.   
\end{definition}

\begin{definition}[Johnson Scheme]
For $n, r \leq n/2$, let $\jh \equiv \jh_{n,r} \subseteq \rnr$ be the subspace of all set-symmetric matrices. $\jh$ is called the Johnson scheme.
\end{definition}
As we will soon see, $\jh$ is also a commutative algebra. There is a natural basis for the subspace $\jh$:
\begin{definition}[D-Basis]
For $0 \leq \ell \leq r \leq n$, let $D_\ell \equiv D_{n,r,\ell} \in \rnr$ be defined by\footnote{We will often omit the subscripts $n,r$.}
\begin{equation}\label{eq:dfnD}
  D_\ell(I,J) =
  \begin{cases}
    1 & |I \cap J| = \ell\\
0 &\text{otherwise.}
  \end{cases}
\end{equation}
\end{definition}
For example, $D_0$ is the well-studied disjointness matrix. Clearly, $\{D_\ell: 0 \leq \ell \leq r\}$ span the subspace $\jh$. Also, it is easy to check that the $D_\ell$'s and hence all the matrices in $\jh$, commute with one another.

Another important collection of matrices that come up naturally while studying PSD'ness of set-symmetric matrices is the following which gives a basis of PSD matrices for the Johnson scheme.
\begin{definition}[P-Basis]
For $0 \leq t \leq r$, let $P_t \equiv P_{n,r,t} \in \rnr$ be defined by\footnote{We will often omit the subscripts $n,r$.}
$$P_t(I,J) = \binom{|I \cap J|}{t}.$$
Equivalently, for $T \subseteq [n]$, if we let $P_T$ be the PSD rank one matrix 
$$P_T =  \mathds{1}\left(\{I: I \subseteq [n], I \supseteq T\}\right) \cdot \mathds{1}\bkets{\braces{I: I \subseteq [n], I \supseteq T}}^{\dagger},$$
then
\begin{equation}
  \label{eq:dfnP}
  P_t = \sum_{T: T \subseteq [n], |T| = t} P_T.
\end{equation}   
 \end{definition}

The equivalence of the above two definitions follows from a simple calculation: there is a non-zero contribution to $(I,J)$'th entry from the $T$'th summand from \eref{eq:dfnP} if and only if $T \subseteq I \cap J$. Clearly, $P_t \succeq 0$ for $0 \leq t \leq r$. We will exploit this relation repeatedly by expressing matrices in $\jh$ as linear combinations of $P_t$'s. The following elementary claim relates the two bases $((D_\ell))$ and $((P_t))$ for fixed $n,r$.
\begin{claim}\label{clm:dtop}
For fixed $n,r$, the following relations hold: 
  \begin{enumerate}
  \item For $0 \leq t \leq r$, $P_t = \sum_{\ell = t}^r \binom{\ell}{t} D_\ell$.
\item For $0 \leq \ell \leq r$, $D_\ell = \sum_{t = \ell}^r (-1)^{t-\ell} \binom{t}{\ell} P_t$.
  \end{enumerate}
\end{claim}
\begin{proof}
The first relation follows immediately from the definition of $P_t$. The second relation follows from inverting the set of equations given in (1).
\end{proof}

The main nontrivial result from the theory of association schemes we use is the following characterization of the eigenspaces of matrices in $\jh$. The starting point for these characterizations is the fact that matrices in $\jh$ commute with one another and hence are simultaneously diagonalizable. We refer the reader to Section 7.4 in \cite{GodsilNotes2} (the matrices $P_t$ in our notation correspond to matrices $C_t$ in \cite{GodsilNotes2}) for the proofs of these results.
\begin{lemma}\label{lm:eigenschemes}
Fix $n, r \leq n/2$ and let $\jh\equiv \jh(n,r)$ be the Johnson scheme. Then, for $P_t$ as defined by \eref{eq:dfnP}, there exist subspaces $V_0, V_1,\ldots,V_r \in \reals^{\bnr}$ that are orthogonal to one another such that:
\begin{enumerate}
\item $V_0,\ldots,V_r$ are eigenspaces for $\{P_t: 0\leq t \leq r\}$ and consequently for all matrices in $\jh$.
\item For $0 \leq j \leq r$, $dim(V_j) = \binom{n}{j} - \binom{n}{j-1}$. 
\item For any matrix $Q \in \jh$, let $\lambda_j(Q)$ denote the eigenvalue of $Q$ within the eigenspace $V_j$. Then, 
  \begin{equation}  
    \label{eq:eigenvaluept}
    \lambda_j(P_t) =
    \begin{cases}
      \binom{n-t-j}{r-t} \cdot \binom{r-j}{t-j} & j \leq t\\
      0 & j > t
    \end{cases}.
  \end{equation}
\end{enumerate}
\end{lemma}

The above lemma helps us estimate the eigenvalues of any matrix in $Q \in \jh$ if we can write $Q$ as a linear combination of the $P_t$'s or $D_\ell$'s. To this end, we shall also use the following estimate on the eigenvalues of such linear combinations.

\begin{lemma}\label{lm:qupperbound}
Let $Q = \sum_\ell \alpha_\ell D_\ell \in \jh(n,r)$, and $\beta_t = \sum_{\ell \leq t} \binom{t}{\ell}\alpha_\ell$, where $\alpha_\ell \geq 0$. Then, for $0 \leq j \leq r$,
$$\lambda_j(Q) \leq  \sum_{t \geq j} \beta_t \cdot \binom{n-t-j}{r-t} \cdot \binom{r-j}{t-j}.$$
\end{lemma}
\begin{proof}
By \clref{clm:dtop},  
\begin{multline*}
  \sum_\ell \alpha_\ell D_\ell = \sum_\ell \alpha_\ell \left(\sum_{t \geq \ell} (-1)^{t-\ell} \binom{t}{\ell}P_t\right) = \sum_t P_t \left(\sum_{\ell \leq t} (-1)^{t-\ell} \binom{t}{\ell} \alpha_\ell\right)\\
\preceq \sum_t P_t \left(\sum_{\ell \leq t}  \binom{t}{\ell} \alpha_\ell\right) = \sum_t \beta_t P_t.
\end{multline*}
Therefore, as $Q$ and $P_t$'s have common eigenspaces, by \lref{lm:eigenschemes},
$$ \lambda_j(Q) \leq  \lambda_j\bkets{\sum_t \beta_t P_t} \leq \sum_t \beta_t \lambda_j(P_t) = \sum_{t \geq j} \beta_t \cdot \binom{n-t-j}{r-t} \cdot \binom{r-j}{t-j}.$$ 
\end{proof}



\section{PSD'ness of the expectation matrices}\label{sec:psdexpectation}
In the section we show that if $r$ is not too large then the expectation matrix $E = \ex[M']$ is PSD with high minimal eigenvalue. As a warmup, we first show that the expectation matrix $E_M = \ex[M]$ is PSD. We start by writing down $E_M$.

\begin{claim}
For $I,J \in \binom{n}{r}$, and $E_M = \ex[M]$, 
\begin{equation}
  \label{eq:exm}
  E_M(I,J) = \binom{n-|I\cup J|}{2r-|I \cup J|} \cdot \frac{\binom{k}{|I \cup J|}}{\binom{2r}{|I\cup J|}} \cdot 2^{-\binom{2r}{2}}.
\end{equation}
\end{claim}
\begin{proof}
The claim follows from observing that for all $I$ and $J$, $\ex[deg_G(I \cup J)] = \binom{n-|I \cup J|}{2r-|I \cup J|} \cdot  2^{-\binom{2r}{2}}$. To see this, note that for all $I$ and $J$ there are $\binom{n-|I \cup J|}{2r-|I \cup J|}$ sets of size $2r$ containing $I \cup J$ and each is a clique with probability $2^{-\binom{2r}{2}}$.
\end{proof}

The expectation matrix above is just a scalar multiple of $\m_{Gr}$ (viewed as a matrix) as defined in \eref{eq:grigoriev}. Therefore, by \tref{th:grigknapsack}, $E_M$ as defined above is PSD for $r < \min(k,n-k)$. We give a simpler proof of this claim here for the case when $r \leq \min(\frac{k}{2},n-k)$.

\begin{theorem}\label{th:knapsack}
The matrix $E_M$ is positive definite for $r \leq \min(\frac{k}{2},n-k)$.   
\end{theorem}
\begin{proof}
We will show this by writing $E_M$ as a suitable positive linear combination of the PSD matrices $P_t$'s from \sref{sec:jscheme}. More concretely, for any $\alpha_0,\ldots,\alpha_t > 0$, we have
$$0 \prec \sum_t \alpha_t P_t = \sum_{\ell = 0}^r \bkets{\sum_{t=0}^\ell \alpha_t \binom{\ell}{t}} D_\ell.$$
Now, let $e_\ell = E_M(I,J)$ for any $I$ and $J$ with $|I \cup J| = 2r-\ell$, i.e.,
$$e_\ell = 2^{-\binom{2r}{2}} \cdot \binom{n-2r+\ell}{\ell} \cdot \frac{\binom{k}{2r-\ell}}{\binom{2r}{2r-\ell}}.$$
Then, $E_M = \sum_{\ell = 0}^r e_\ell D_\ell$. Therefore, we will be done if we can find $\alpha_t$'s such that for every $0 \leq \ell \leq r$, $ e_\ell = \sum_{t=0}^\ell \alpha_t \binom{\ell}{t}$. By examining the first values of $\ell$, it is easy to guess what the $\alpha_t$ should be. 
First observe that $e_\ell = e_0 \cdot \binom{n-2r+\ell}{\ell}/\binom{k-2r+\ell}{\ell}$ and let $\alpha_t = e_0 \binom{n-k}{t}/\binom{k-2r+t}{t}$. Then,
\begin{align*}
 e_0 \binom{n-2r+\ell}{\ell} &= e_0 \sum_{t=0}^\ell \binom{n-k}{t} \cdot \binom{k-2r+\ell}{\ell-t}\\
&= \sum_{t=0}^\ell  \alpha_t \cdot \binom{k-2r+t}{t} \binom{k-2r+\ell}{\ell-t}\\
&= \sum_{t=0}^\ell \alpha_t \cdot \binom{\ell}{t} \cdot \binom{k-2r+\ell}{\ell}.
\end{align*}
Therefore, $e_\ell = \sum_t \binom{\ell}{t} \alpha_t$ and the lemma now follows:
\begin{align*}
  E_M = \sum_{\ell = 0}^r e_\ell D_\ell = \sum_{\ell =0}^r \left(\sum_{t=0}^\ell \alpha_t \binom{\ell}{t}\right) D_\ell  \succeq \alpha_r \mathds{I}.
\end{align*}
\end{proof}

\subsection{PSD'ness of $E$}
Now that we have shown that $E_M$ is positive definite when $r$ is not too large, we use similar ideas to analyze $E = \ex[M']$, the expectation matrix we will actually be using. We begin by writing down $E$.
\begin{claim}
For $I,J \in \binom{n}{r}$, and $E = \ex[M']$, 
\begin{equation}
  \label{eq:actualexm}
  E(I,J) = \binom{n-|I\cup J|}{2r-|I \cup J|} \cdot \frac{\binom{k}{|I \cup J|}}{\binom{2r}{|I\cup J|}} \cdot 2^{-{r^2} - \binom{|I \cap J|}{2}}.
\end{equation}
\end{claim}
\begin{proof}
The claim follows from observing that for all $I$ and $J$, conditioned on the edges in $\mathcal{E}(I)$ and $\mathcal{E}(J)$ being present, $\ex[deg_G(I \cup J)] = \binom{n-|I \cup J|}{2r-|I \cup J|} \cdot 2^{-\binom{2r}{2} + \binom{|I|}{2} + \binom{|J|}{2} - \binom{|I \cap J|}{2}} = \binom{n-|I \cup J|}{2r-|I \cup J|} \cdot 2^{-r^2 - \binom{|I \cap J|}{2}}$. To see this, note that for all $I$ and $J$ there are $\binom{n-|I \cup J|}{2r-|I \cup J|}$ sets of size $2r$ containing $I \cup J$, and conditioned on the edges in $\mathcal{E}(I)$ and $\mathcal{E}(J)$ being present, each is a clique with probability $2^{-\binom{2r}{2} + \binom{|I|}{2} + \binom{|J|}{2} - \binom{|I \cap J|}{2}}$. Now note that $|I| = |J| = r$ and $-\binom{2r}{2} + 2\binom{r}{2} = -(2r^2 - r) + (r^2 - r) = -r^2$ so 
$-\binom{2r}{2} + \binom{|I|}{2} + \binom{|J|}{2} - \binom{|I \cap J|}{2} = -r^2 - \binom{|I \cap J|}{2}$
\end{proof}

\begin{lemma}\label{lm:exmineigenvalue}
If $k < \frac{n - 2r}{3r \cdot 2^{r-1}}$ and $r \leq \frac{k}{2}$ then $E$ is PSD with minimal eigenvalue $2^{-O(r^2)}{k^r}{n^r}$
\end{lemma}
\begin{proof}
By \eref{eq:actualexm}, $E = \sum_\ell e_{\ell} D_\ell$, where $e_{\ell} = \binom{n-2r+l}{l} \cdot \frac{\binom{k}{2r - \ell}}{\binom{2r}{2r - \ell}} \cdot 2^{-{r^2} - \binom{\ell}{2}}$. We next express $E$ as a linear combination of $P_t$'s: $E = \sum_t \alpha_t P_t$. By \clref{clm:dtop}, $D_\ell = \sum_{t = \ell}^r (-1)^{t-\ell} \binom{t}{\ell} P_t$ so
$$\alpha_t = \sum_{\ell = 0}^{t} (-1)^{t - \ell} \binom{t}{\ell} e_{\ell}.$$
Now note that for all $l \geq 1$, $e_{\ell} = \frac{n - 2r + \ell}{\ell}\cdot\frac{\ell}{k - 2r + \ell}\cdot{2^{\ell-l}}\cdot{e_{\ell-1}} = \frac{n - 2r + \ell}{2^{1-l}(k - 2r + \ell)}\cdot{e_{\ell-1}}$. If 
$k < \frac{n - 2r}{3r \cdot 2^{r-1}}$ then the terms in the sum for $\alpha_t$ increase geometrically by a factor of at least 3 and the sum will therefore be dominated by the last term. In particular, $\alpha_t \geq \frac{e_t}{2}$. Thus, $\alpha_t > 0$ for all $t \in [0,r]$ and $$\alpha_r \geq \frac{e_r}{2} = \frac{1}{2} \cdot \binom{n-r}{r} \cdot \frac{\binom{k}{r}}{\binom{2r}{r}} \cdot 2^{-{r^2} - \binom{r}{2}} = 2^{-O(r^2)}{k^r}{n^r}$$ Since the $P_t$'s are PSD and $P_r = I$, $E$ is PSD with minimal eigenvalue $2^{-O(r^2)}{k^r}{n^r}$, as needed.
\end{proof}

\section{PSD'ness of dual certificate}\label{sec:psdmaintech}
We are now ready to prove our main result, \tref{th:main}, with the aid of several technical results whose proof is deferred to \sref{sec:localrand} and \sref{sec:degreevariancecalculations}. We prove \tref{th:main} by showing that the matrix $M$ will be PSD with high probability (\tref{th:mainpsd}). In turn, we show that $M$ is PSD with high probability with our main technical lemma, which says that $M'$ is PSD with high probability (this is sufficient by \lref{lm:MprimePSDtoMPSD}).
\begin{lemma}[Main Technical Lemma]\label{lm:maintech}
For $c$ a sufficiently large constant the following holds. The matrix $M' \in \reals^{\binom{[n]}{r} \times \binom{[n]}{r}}$ defined by \eref{eq:mr} is positive definite with high probability, for $k < 2^{-cr} (\sqrt{n}/\log n)^{1/r}$.
\end{lemma}

To prove \lref{lm:maintech}, we first decompose $M'$ as $M' = E + L + \Delta$ in \sref{sec:psdmmr}. We then analyze $L$ and $\Delta$ in \sref{sec:normboundL} and \sref{sec:normboundDelta} respectively. We put all the pieces together to show the PSD'ness of $M'$ in \sref{sec:wrapup}. 

For the remainder of this section, we shall use the following additional notations: 
\begin{itemize}
\item For $0 \leq i \leq r$, let 
\begin{equation}\label{eq:defalpha}
\alpha(i) = \frac{\binom{k}{2r-i}}{\binom{2r}{2r-i}} \cdot \binom{n-2r+i}{i} \cdot 2^{-r^2 - \binom{i}{2}}.
\end{equation}
\ignore{\item For $0 \leq i \leq r$, let 
\begin{equation}
\beta'(i) = \beta(i)/2^{r^2 + \binom{i}{2}} = \frac{\binom{k}{2r-i}}{\binom{2r}{2r-i}} \cdot 2^{r^2 + \binom{i}{2}}.
\end{equation}}

\item For $0 \leq i \leq r$, let $p(i) = 2^{-(r-i)^2}$. Then, for $I,J \in \bnr$ with $|I \cap J| = i$, $p(i)$ is the probability that $\mathcal{E}(I \cup J) \setminus (\mathcal{E}(I) \cup \mathcal{E}(J)) \subseteq G$.
\item In the following we will adopt the convention that $I,J,K$ denote elements of $\binom{[n]}{r}$ and $T,T'$ denote elements of $\binom{[n]}{2r}$. 
\item All matrices considered below will be over $\reals^{\binom{[n]}{r} \times \binom{[n]}{r}}$ unless otherwise specified.
\item We write $A \approx_r B$ if there exist constants $c, C$ such that $c^{r^2} B \leq A \leq C^{r^2} B$.
\end{itemize}

\subsection{Decomposition of $M'$}\label{sec:psdmmr}
For the proof of \lref{lm:maintech} we decompose the matrix $M'$ as $M' = E + L + \Delta$, where (a) $E = \ex[M']$ is the expectation matrix; (b) $L$ will be a ``local'' random matrix such that for sets $I,J$, $L(I,J)$ only depends on the edges between the vertices of $I \cup J$ and (c) $\Delta$ is a ``global'' error matrix whose entries are small in magnitude.

To this end, first observe that by \eref{eq:actualexm}, for $E \equiv \ex[M']$, 
\begin{equation}
  \label{eq:Evalue}
  E(I,J) = \alpha(|I \cap J|).
\end{equation}

Now, define $L \in \rnr$ as follows: for $I,J \in \bnr$,
\begin{equation}
  \label{eq:defL}
L(I,J) =
\begin{cases}
 \alpha(|I \cap J|) \cdot \frac{1-p(|I \cap J|)}{p(|I \cap J|)} &\text{ if $\mathcal{E}(I \cup J) \setminus (\mathcal{E}(I) \cup \mathcal{E}(J)) \subseteq G$}\\
- \alpha(|I \cap J|) &\text{otherwise}
\end{cases}.
\end{equation}
Finally, define $\Delta = M' - E - L$. We have already shown in \sref{sec:psdexpectation} that $E$ is PSD with minimal eigenvalue $2^{-O(r^2)}{k^r}{n^r}$. There are now two remaining modular steps in the proof:
\begin{enumerate}
\item We show that $\|L\|$ is $2^{O(r^2)}k^{2r}n^{r-1/2}\log{n}$ by exploiting the recursive structure of the matrix $L$ and some careful trace calculations. This is the most technically intensive part of the proof. 
\item We then show that $\|\Delta\|$ is $2^{O(r^2)}k^{2r}n^{r-1/2}\log{n}$. This is done by first showing that each entry of $\Delta$ is small in magnitude and using \lref{lm:Gershgorin}.
\end{enumerate}

The next two subsections address these two steps with the corresponding technical elements dealt with in \sref{sec:localrand} and \sref{sec:degreevariancecalculations} respectively.

\subsection{Bounding the norm of the locally random matrix $L$}\label{sec:normboundL}
In this subsection, we bound the norm of the matrix $L$.

\begin{lemma}\label{lm:Lnorm}
For some constant $C > 0$, with probability at least $1 - 1/n$ over the random graph $G$, 
$$\|L\| \leq O(1) \cdot 2^{C r^2} \cdot k^{2r} \cdot n^r \cdot \frac{\log n}{\sqrt{n}}.$$
\end{lemma}

We will prove the lemma by further decomposing $L$ according to the intersection sizes of the indexing sets and using the recursive structure of the matrix $M'$. To this end, we define the following closely related locally-random matrix. For $a \in [r]$, let $R_a \in \reals^{\binom{n}{a} \times \binom{n}{a}}$ be the matrix supported only on disjoint sets and defined as follows: for $V,W \in \binom{[n]}{a}$, 
\begin{equation}
  \label{eq:defR}
  R_a(V,W) =
  \begin{cases}
    2^{a^2} - 1 &\text{ if $V \cap W = \emptyset$ and $\{\{v,w\}: v \in V, w \in W\} \subseteq G$}\\
    -1 &\text{ if $V \cap W = \emptyset$ and $\{\{v,w\}: v \in V, w \in W\} \not\subseteq G$}\\
0 &\text{ if $V \cap W \neq \emptyset$} 
  \end{cases}.
\end{equation}
In other words, for disjoint $V,W \in \binom{[n]}{a}$ the $R_a(V,W)$'th entry is essentially (up to a constant multiple) a shift of the indicator random variable which is $1$ if all edges in $V \times W$ are in $G$ and $0$ otherwise. 

Note that $\ex[R_a] = 0$. The following technical claim proved in Section \ref{sec:localrand} bounds the norm of $R_a$. The proof relies on computing the trace of powers of $R_a$.
\begin{claim}[See Section \ref{sec:localrand}]\label{clm:maintechnicalL}
If $n \geq 100$, for all $\epsilon \in (0,1)$, 
$\pr\left[||R_a|| > 2^{a^2+2a+2}\ln{(\frac{n}{\epsilon})}n^{a - \frac{1}{2}}\right] < \epsilon$.
\ignore{
For some constant $C > 0$, for all $\epsilon \in (0,1)$,
$$\pr\left[ \|R_a\| > C 2^{a^2} \cdot \log(n/\epsilon) \cdot n^{a-1/2}\right] < \epsilon.$$}
\end{claim}
Note that $2^{a^2} n^a$ is an easy bound for $\|R_a\|$ (each entry of the matrix is at most $2^{a^2}$ in magnitude); the main advantage of the claim is the multiplicative $n^{-1/2}$ factor.

In the remainder of this section we use the recursive structure of the matrix $L$ to prove \clref{lm:Lnorm} assuming the above claim. We first introduce some notation:
\begin{itemize}
\item For a matrix $X \in \reals^{\binom{[n]}{r_1} \times \binom{[n]}{r_2}}$, and $0 \leq i \leq \min{\{r_1,r_2\}}$, let $X^i \in \reals^{\binom{[n]}{r_1} \times \binom{[n]}{r_2}}$ be the matrix such that $X^i(I,J) = X(I,J)$ if $|I \cap J| = i$ and $0$ otherwise\footnote{For this paper, we will only use the case where $r_1 = r_2 = r$. We put in this extra generality with an eye towards future work.}.
\item For a matrix $X \in \reals^{\binom{[n]}{r_1-i} \times \binom{[n]}{r_2-i}}$, let $X^{(i)} \in \reals^{\binom{[n]}{r_1} \times \binom{[n]}{r_2}}$, be defined as follows:
\begin{equation}
  \label{eq:defextension}
  X^{(i)}(I,J) =
  \begin{cases}
X(I\setminus (I \cap J), J \setminus (I \cap J)) &\text{ if $|I \cap J| = i$}\\
0 &\text{ otherwise}
  \end{cases}.
\end{equation}
\end{itemize}
\ignore{
The following claim immediately implies \lref{lm:Lnorm} as $L = \sum_{i=0}^r L^i$.
\begin{claim}\label{clm:Lnorm1}
For some constant $C > 0$, with probability at least $1-1/n^2$ over the random graph $G$, for every $0 \leq i \leq r$, 
$$\|L^i\| \leq O(1) \cdot 2^{C r^2} \cdot k^{2r-i} \cdot n^{r} \cdot \frac{\log n}{\sqrt{n}}.$$
\end{claim}}

The next claim relates the norms of ``lifts'' of matrix $R$, $R^{(i)}$. Conceptually, bounding the norms of matrices with non-zero entries on intersecting indexing sets are reduced to that of the disjoint case. Note that the requirement $R=R^0$ exactly captures the latter. 
\begin{lemma}\label{lm:recursiveL}
For $0 \leq i \leq \min{\{r_1,r_2\}}$ and $R \in \reals^{\binom{[n]}{r_1-i} \times \binom{[n]}{r_2-i}}$, if 
$R = R^0$ then $\|R^{(i)}\| \leq \binom{r_1}{i}\binom{r_2}{i} \cdot \|R\|$.
\end{lemma}
\begin{proof}
We partition the entries of $R^{(i)}$ as follows. 
\begin{definition}
For any $X,Y,K$ such that $X \subseteq [1,r_1]$, $Y \subseteq [1,r_2]$, and $K \subseteq V(G)$ where $|K| = |X| = |Y| = i$, 
let $R_{X,Y,K}^{(i)}$ be the matrix such that the following is true: 
\begin{enumerate}
\item $R_{X,Y,K}^{(i)}(I,J) = R^{(i)}(I,J) = R(I \setminus K, J \setminus K)$ if 
$K = \{i_x: x \in X\} = \{j_y: y \in Y\}$ where $i_1, \cdots, i_{r_1}$ are the elements of $I$ in increasing order and $j_1, \cdots, j_{r_2}$ are the elements of $J$ in increasing order.
\item $R_{X,Y,K}^{(i)}(I,J) = 0$ otherwise.
\end{enumerate}
\end{definition}
\begin{proposition}\label{decompositionpropone}
For all $X,Y,K$, $||R_{X,Y,K}^{(i)}|| \leq ||R||$.
\end{proposition}
\begin{proof}
The nonzero part of $R_{X,Y,K}^{(i)}$ can be viewed as a submatrix of $R$, so it cannot have 
larger induced norm than $R$.
\end{proof}
\begin{proposition}\label{decompositionproptwo}
$R^{(i)} = \sum_{X,Y,K}{R_{X,Y,K}^{(i)}}$.
\end{proposition}
\begin{proof}
If $R^{(i)}(I,J) = 0$ then $\sum_{X,Y,K}{R_{X,Y,K}^{(i)}(I,J)} = 0$. If $R^{(i)}(I,J) \neq 0$ then 
$|I \cap J| = i$. This implies that $K = \{i_x: x \in X\} = \{j_y: y \in Y\}$ if and only if 
$K = I \cap J$, $X$ is the set of indices of $K$ in $I$, and $Y$ is the set of indices of $K$ in $J$, which 
happens for precisely one $X,Y,K$. Thus, $R_{X,Y,K}^{(i)}(I,J) = R^{(i)}(I,J)$ for precisely one $I,J,X$ 
and is $0$ otherwise, so $R^{(i)} = \sum_{X,Y,K}{R_{X,Y,K}^{(i)}}$, as needed.
\end{proof}
\begin{proposition}\label{decompositioncorone}
$||R^{(i)}|| \leq \sum_{X,Y}{||\sum_{K}{R_{X,Y,K}^{(i)}}||}$.
\end{proposition}
\begin{proposition}\label{decompositionpropthree}
If $K_1$,$K_2$ are distinct subsets of $V(G)$ of size $x$, $R_{X,Y,K_1}^{(i)}(I_1,J_1) \neq 0$, and 
$R_{X,Y,K_2}^{(i)}(I_2,J_2) \neq 0$ then $I_1 \neq I_2$ and $J_1 \neq J_2$.
\end{proposition}
\begin{proof}
Assume that $I_1 = I_2 = I$ and let $i_1, \cdots, i_{r_1}$ be the elements of $I$ in increasing order. Then 
$K_1 = \{i_x: x \in X\} = K_2$. Contradiction. Following similar logic, we cannot have that $J_1 = J_2$ either. 
\end{proof}
\begin{proposition}\label{decompositioncortwo}
For any $X,Y \subseteq [1,n]$, $||\sum_{K}{R_{X,Y,K}^{(i)}}|| \leq ||R||$.
\end{proposition}
\begin{proof}
Note that we can permute the rows and columns of a matrix without affecting its induced norm. By Proposition 
\ref{decompositionpropthree}, we can permute the rows and columns of $\sum_{K}{R_{X,Y,K}^{(i)}}$ to put it into block 
form where each block is the nonzero part of $R_{X,Y,K}^{(i)}$ for some $K$. For a matrix in block form, its norm is 
the maximum of the norms of the individual blocks, which by Proposition \ref{decompositionpropone} is at most $||R||$, as needed.
\end{proof}
With these results, Lemma \ref{lm:recursiveL} follows immediately. Plugging in Proposition \ref{decompositioncortwo} to 
Proposition \ref{decompositioncorone} gives 
$||R^{(i)}|| \leq \sum_{X,Y}{||\sum_{K}{R_{X,Y,K}^{(i)}}||} \leq \sum_{X,Y}{||R||} \leq 
{{r_1} \choose i}{{r_2} \choose i}||R||$, as needed.
\end{proof}

We now use the above statements to prove \lref{lm:Lnorm}.
\begin{proof}[Proof of \lref{lm:Lnorm}]
We claim that for $0 \leq i \leq r$, and $\alpha_i$ as in \eref{eq:defalpha}
\begin{equation}
  \label{eq:Lnorm11}
L^i = \alpha_i \cdot R_{r-i}^{(i)}.  
\end{equation}

To see the above, fix $I, J \in \bnr$ with $|I \cap J| = i$ and let $V = I \setminus (I \cap J)$, $W = J \setminus (I \cap J)$. Observe that 
$$\mathcal{E}(I \cup J) \setminus (\mathcal{E}(I) \cup \mathcal{E}(J)) = \{\{v,w\}: v \in V, w \in W\}.$$

We cosider two cases as in the definition of $L$. 

{\bf Case 1}. $\mathcal{E}(I \cup J) \setminus (\mathcal{E}(I) \cup \mathcal{E}(J)) \subseteq G$. Then, $R_{r-i}^{(i)}(I,J) = R_{r-i}(V,W) = 2^{(r-i)^2 - 1} = (1-p(i))/p(i)$. \eref{eq:Lnorm11} now follows from the first case of the definition of $L$.

{\bf Case 2}. $\mathcal{E}(I \cup J) \setminus (\mathcal{E}(I) \cup \mathcal{E}(J)) \not\subseteq G$. Then, $R_{r-i}^{(i)}(I,J) = R_{r-i}(V, W) = -1$. \eref{eq:Lnorm11} now follows from the second case of the definition of $L$. 

Therefore, by Claim \ref{clm:maintechnicalL}, Lemma \ref{lm:recursiveL} and \eref{eq:Lnorm11},
$$\|L^i\| \leq O(1) \cdot 2^{C r^2} \cdot k^{2r-i} \cdot n^{r} \cdot \frac{\log n}{\sqrt{n}}.$$
The lemma now follows as $L = \sum_{i=0}^r L^i$. 
\end{proof}

\subsection{Bounding the norm of the global error matrix $\Delta$}\label{sec:normboundDelta}
The main claim of this subsection is the following bound on the spectral norm of $\Delta$.
\begin{lemma}\label{lm:Deltanorm}
For $n > C 2^{4r^2}$, with probability at least $1 - 1/n$ over the random graph $G$, 
$$\|\Delta\| \leq 2^{Cr^2} \cdot k^{2r} \cdot n^r \cdot \frac{\log n}{\sqrt{n}}.$$
\end{lemma}

The proof relies on the following bound on the individual entries of $\Delta$. 
\begin{lemma}\label{lm:Deltaentrywise}
For some universal constant $C$, and $n > C 2^{4r^2}$, with probability at least $1 - 1/n$ over the random graph $G$, for all $I,J \in \bnr$, with $i = |I \cap J|$,
$$|\Delta(I,J)| \leq 2^{C r^2} \cdot k^{2r-i} \cdot n^i \cdot \frac{\log n}{\sqrt{n}}.$$
\end{lemma}
Before proving the lemma, we first use it to bound $\|\Delta\|$.
\begin{proof}[Proof of \lref{lm:Deltanorm}]
Suppose that the conclusion of the previous lemma holds. Then, for any $I \in \bnr$,
\begin{align*}
\|\Delta \|\leq  \sum_J |\Delta(I,J)| &= \sum_{i=0}^r \sum_{J: | I \cap J| = i} |\Delta(I,J)|\\
&\leq \frac{2^{C r^2} k^{2r} (\log n)}{\sqrt{n}} \sum_{i=0}^r \sum_{J: |I \cap J| = i} (n/k)^i\\
&\leq \frac{ 2^{Cr^2} k^{2r} (\log n)}{\sqrt{n}} \sum_{i=0}^r (n/k)^i\,2^r n^{r-i}\\
&\leq \frac{ 2^{Cr^2} k^{2r} (\log n) n^r}{\sqrt{n}}.
\end{align*}
The lemma now follows from the above bound and \lref{lm:Gershgorin}.
\end{proof}

\begin{proof}[Proof of \lref{lm:Deltaentrywise}]
Fix sets $I,J$ with $|I \cap J| = i$. Let $\mathcal{A}$ be the event that $\mathcal{E}(I \cup J) \setminus (\mathcal{E}(I) \cup \mathcal{E}(J)) \subseteq G$. 

Then, by the second case of \eref{eq:defL}, conditioned on $\neg\mathcal{A}$ we have $\Delta(I,J) = 0$. Thus, the claim holds trivially in this case. In the following we condition on $\mathcal{A}$. Observe that
$$\ex[M'(I,J) \mid \mathcal{A}] = E(I,J)/\pr[\mathcal{A}] = E(I,J)/p(i).$$

We next use the following claim that $\deg_G(I\cup J)$ is concentrated around its mean when conditioned on $I \cup J$ being a clique. At a high level, this follows from the fact that conditioned on $I \cup J$ being a clique, $\deg_G(I \cup J)$ can be written as a (structured) low-degree polynomial in the indicator variables of the edges not in $I \cup J$ with small variance. We defer the proof to the appendix.

\begin{claim}[See \tref{th:degreeconcentrationfinal} of the appendix]\label{clm:degreebound}
For some constant $C > 0$,
$$\pr\left[\left| \deg_G(I \cup J) - 2^{-\binom{2r}{2} + \binom{2r-i}{2}} \cdot \binom{n-2r+i}{i}  \right| > 2 (\ln(C/\epsilon))^2 n^{i-1/2} \mid (I \cup J \text{ a clique})\right] < \epsilon.$$
\end{claim}

As a consequence of the above claim we also get concentration for $M'(I,J) \mid \mathcal{A}$. This is because $M'(I,J) \mid \mathcal{A}$ is identically distributed as $M(I,J) \mid (I \cup J\text{ a clique})$. Therefore, taking $\epsilon = 1/n^{2r+1}$ and applying a union bound over all sets $I,J$ we get that with probability at least $1-1/n$, for all $I,J$ such that $\mathcal{E}(I \cup J) \setminus (\mathcal{E}(I) \cup \mathcal{E}(J)) \subseteq G$, and $|I \cap J| = i$,

$$\left| M_r'(I,J) - \beta(i) 2^{-\binom{2r}{2} + \binom{2r-i}{2}} \cdot \binom{n-2r+i}{i} \right| < C r 2^{2r^2} \cdot k^{2r-i} \cdot (\log n) \cdot n^{i-1/2}.$$

Finally, observe that 
$$\beta(i) 2^{-\binom{2r}{2} + \binom{2r-i}{2}} \cdot \binom{n-2r+i}{i} = \alpha(i)/p(i),$$
and conditioned on $\mathcal{A}$, $\Delta(I,J) = M'(I,J) - \alpha(|I \cap J|)/p(|I \cap J|)$.  The lemma now follows by combining the above two bounds.
\end{proof}

\subsection{Putting things together}\label{sec:wrapup}
We now prove \lref{lm:maintech} and use it to prove our main results.
\begin{proof}[Proof of \lref{lm:maintech}]
By \lref{lm:exmineigenvalue}, we have that $E \succeq 2^{-Cr^2} k^r n^r \mathbb{I}$. Therefore, by \lref{lm:Lnorm} and \lref{lm:Deltanorm}, with probability at least $1-2/n$, the least eigenvalue of $M'$ is at least
$$2^{-Cr^2} k^r n^r - 2^{O(r^2)} k^{2r} n^r \frac{\log n}{\sqrt{n}} = k^r n^r \left(2^{-O(r^2)} - \frac{2^{O(r^2)} k^r (\log n)}{\sqrt{n}}\right) \geq 0,$$
for $k$ as in the statement of the lemma for a sufficiently big constant $c$.  
\end{proof}

We bring the arguments from previous sections together to prove our main results \tref{th:mainpsd} and \tref{th:main}.

\begin{proof}[Proof of \tref{th:mainpsd}]
Follows immediately from \lref{lm:MprimePSDtoMPSD} and \lref{lm:maintech}.
\end{proof}

\begin{proof}[Proof of \tref{th:main}] 
Follows immediately from \lref{lm:dualcert}, \clref{clm:dualsol} and \tref{th:mainpsd}.
\end{proof}

Theorems \ref{th:mainhierarchy} and \ref{cor:mainhierarchy} follow immediately from our $\psd$-refutation lower bound using standard arguments. We defer these to the appendix.


\section{Bounding norms of locally random matrices}\label{sec:localrand}
In this section we shall develop tools for bounding the norms of \emph{locally random} matrices (recall their informal definition from \sref{sec:outline} and more formal one in \sref{sec:normboundL}) associated with random graphs $G \lfta G(n,1/2)$, proving \clref{clm:maintechnicalL}. The idea behind our bounds is to use the \emph{trace method}. Recall the trace method: for any matrix $M$, for any positive integer $q$, $||M|| \leq \sqrt[2q]{tr(({M^T}M)^{q})}$ so we can probabilistically bound $||M||$ by bounding $\ex\sbkets{tr(({M^T}M)^{q})}$.

Going back to \clref{clm:maintechnicalL} let us first look at the special case of $a=1$ to gain some intuition. In this case, the entries of $R_1$ are (essentially) independent, and so the trace method is easy to apply. More precisely, $R_1$ is a symmetric random matrix with zeros on the diagonal and the entries in the upper diagonal taking independent uniformly random $\pm{1}$ values. It is well known that $\|R_1\| = O(\sqrt{n})$ in this case (see \cite{Vershynin06} for instance). One can also prove the bound by the trace method as follows. We have that  
$$\ex\sbkets{tr(({{R_1}^T}{R_1})^{q})} = \ex\sbkets{tr({R_1}^{2q})} = \sum_{i_1,\cdots,i_{2q}}{\ex\sbkets{\prod_{j=1}^{2q}{{R_1}({i_j},i_{j+1})}}},$$ where $i_{2q+1} = i_1$. We can then look at which products $\prod_{j=1}^{2q}{{R_1}({i_j},i_{j+1})}$ have expectation $0$. 

Since each individual ${R_1}({i_j},i_{j+1})$ is an independent $\pm{1}$ random variable with expectation $0$, a term in the summation $\ex\sbkets{\prod_{j=1}^{2q}{{R_1}({i_j},i_{j+1})}} = 0$ unless every ${R_1}({i_j},i_{j+1})$ appears an even number of times in the product. Thus, the vast majority of the terms $\ex\sbkets{\prod_{j=1}^{2q}{{R_1}({i_j},i_{j+1})}}$ are $0$ and we can count the remaining terms to bound $\ex\sbkets{tr(R_1^{2q})}$.


One way to implement the above argument is to first look at terms $\ex\sbkets{\prod_{j=1}^{2q}{{R_1}({i_j},i_{j+1})}}$ which have non-zero expectation and observe that in all such terms, the number of distinct entries in $\{i_j: j =1,\ldots,2q\}$ is at most $q+1$. We can then bound the number of terms with non-zero expected value by the number of possible terms which contain at most $q+1$ distinct elements. This number can be easily bounded by $O((nq)^{q+1})$, and picking $q$ optimally results in showing that with high probability $\|R_1\| = O(\sqrt{n}\log n)$, a near-optimal bound. 


To handle higher $a$'s we first generalize the above argument based on constraint graphs to work with general locally-random matrices. However, unlike for $a=1$, distinct entries of the matrix are now dependent, which significantly complicates the structure of the terms and the associated count of the terms which have non-zero expectation. The rest of the section is devoted to this. While we apply our arguments to the particular locally-random matrices arising in our proof, these techniques should apply more generally to other locally-random matrices.  

\subsection{Constraint graphs}\label{constraintgraphs}
We next state our main technical result which gives us a way to bound traces of high powers of locally random matrices based on the structure of the individual terms. The advantage being that the conditions on the terms will be easier to ascertain in our applications. 

Here we use $V$ rather than $I$ for subsets because we will be viewing the individual elements of each $V$ as vertices. 

\begin{theorem}\label{thm:productgraphtonorm}
Assume that we have values $a, B>0$ and for every positive $q$, we have a function $p(G,2q)$ such that $p(G,2q) \geq 0$ and $p(G,2q)$ can be written in the form 
$$p(G,2q) = \sum_{\{V_1,\ldots,V_{2q}\}}{f(G,\{V_1,\ldots,V_{2q}\})}$$
where the following are true:
\begin{enumerate}
\item $V_j \subseteq V(G)$ and $|V_j| = a$. 
\item For every term $f(G,\{V_1,\ldots,V_{2q}\})$ with non-zero expected value, $|\cup_j V_j| \leq 2aq-qy+z$ for some integers $y$ and $z$ where $1 \leq y \leq 2a$ and $z \geq 0$. 
\item $\ex[f(G,\{V_1,\ldots,V_{2q}\})] \leq B^{2q}$.
\end{enumerate}
Then, if $n \geq 10$, for all $\epsilon \in (0,1)$, 
$$\pr\left[|\min_{q \in \mathbb{Z}^{+}}{\{\sqrt[2q]{p(G,2q)}\}}| > \frac{B}{a!} \cdot \left(2ea\bkets{\frac{\ln(n^z/\epsilon)}{2y} + 1}\right)^y \cdot n^{a-y/2}\right] < \epsilon.$$
\end{theorem}

\begin{remark}
We will use this theorem with two types of functions $p$. When $p(G,2q) = tr(({M^T}M)^q)$ for some matrix $M$ depending on $G$, 
$||M|| \leq \sqrt[2q]{p(G,2q)}$ for all $q > 0$ so this theorem gives us a probabilistic bound on $||M||$. When $p(G,2q) = h(G)^{2q}$ for some function $h$, then $h(G) = \sqrt[2q]{p(G,2q)}$ for all $q > 0$ so this theorem gives us a probabilistic bound on $h(G)$.
\end{remark}
\begin{example}
In the case when $p(G,2q) = tr(R_1^{2q})$, $p(G,2q) = \sum_{i_1,\cdots,i_{2q}}{\prod_{j=1}^{2q}{{R_1}({i_j},i_{j+1})}}$. Each term here has expected value at most $1$ and it is easy to argue that for any term with non-zero expected value, the number of distinct elements is at most $q+1$. Applying \tref{thm:productgraphtonorm} with $y = z = 1$, and $B = 1$ we have that for all $n \geq 10$, and $\epsilon \in (0,1)$, 
$$\pr\left[||R_1|| > 2e\sqrt{n} (\ln{(n/\epsilon)} + 2) \right] < \epsilon$$
This bound is weaker (by a logarithmic factor) than the bounds in e.g. \cite{Vershynin06}, but is sufficient for our purposes.
\end{example}

Before proving the theorem we introduce the concept of {\sl constraint graphs} which are a useful way to visualize our calculations. While the statement of the above theorem does not involve constraint graphs, thinking in terms of constraint graphs is helpful in proving the conditions required to apply the theorem. 

\begin{definition}
Given a family of sets of vertices $\{V_i\}$, we define a corresponding constraint graph $C$ whose vertices are the sets $\{V_i\}$ and there is an edge between $V_i, V_j$, $i \neq j$, if $V_i \cap V_j \neq \emptyset$.
\end{definition}

The above definition is useful because of the following elementary lemma.
\begin{lemma}\label{lm:componentstounion}
For any collection of sets $\{V_1,\ldots,V_{\ell}\}$, if the corresponding constraint graph $C$ has $t$ connected components, then $|\cup_i V_i| < \sum_i |V_i| - \ell + t$.
\end{lemma}
\begin{proof}
Let $V_{i_1},\ldots,V_{i_t}$ belong to the $t$ different connected components of $C$. Now add the remaining elements of $\{V_1,\ldots,V_\ell\}$ so that each new set is adjacent (in $C$) to at least one of the previously added sets (we can do this as the number of connected components is $t$). Then, each such step adding a set $V_i$ can increase the size of the union by at most $|V_i|-1$. Therefore, the size of the union is at most $\sum_i |V_i| - \ell + t$.
\end{proof}

\begin{proof}[Proof of \tref{thm:productgraphtonorm}]
In the following we use $\{V_i\}$ as a short form for $\{V_1,\ldots,V_{2q}\}$. We prove this result by obtaining an upper bound on the number of terms in $p(G,2q) = \sum_{\{V_{i}\}}{f(G,\{V_{i}\})}$ with nonzero expected value. This gives us a probabilistic upper bound for $p(G,2q)$, implying the upper bound on $\min_q{\{\sqrt[2q]{p(G,2q)}\}}$.
\begin{definition}
Define $N(n,a,q,m)$ to be the number of ways to choose subsets $\{V_{i}: i \in [2q]\}$ of $[n]$ such that $|\cup_{i}{V_{i}}| \leq m$ and for all $i$, $|V_{i}| = a$.
\end{definition}
\begin{lemma}\label{tboundlemma}
If $m \leq 2aq$, then 
$$N(n,a,q,m) \leq \bkets{\frac{1}{a!}}^{2q}{\binom{2aq}{2aq-m}}{n^{m}}{{m}^{2aq-m}}$$
\end{lemma}
\begin{proof}
We can choose each ordered $2aq$-tuple $(v_1, \cdots, v_{2aq})$ of elements in $[n]$ which contains at most $m$ distinct elements as follows. 
There must be at least $2aq-m$ elements which are duplicates of other elements, so we can first choose a set $I$ of $2aq-m$ indices such that for all $i \in I$, $v_i = v_j$ for some $j \notin I$. There are $\binom{2aq}{2aq-m}$ choices for $I$. We then choose the elements $\{v_j: j \notin I\}$. There are no restrictions on these elements so there are $n^{m}$ choices for these elements. Finally, we choose the elements $\{v_i: i \in I\}$. To determine each $v_i$ it is sufficient to specify the $j \notin I$ such that $v_i = v_j$. For each $i$ there are $m$ choices for the corresponding $j$, so the number of choices for these elements is at most $m^{2aq-m}$. Putting everything together, the total number of choices is at most $\binom{2aq}{2aq-m}{n^{m}}{{m}^{2aq-m}}$. Now note that since we are choosing subsets $\{V_{i}: i \in [2q]\}$ of $[n]$ rather than one big ordered tuple, the order within each subset does not matter. Thus, there are $(a!)^{2q}$ different ordered tuples which give the same subsets of elements, so the total number of possibilities for the subsets $\{V_{i}\}$ is at most $(a!)^{-2q}{\binom{2aq}{2aq-m}}{n^{m}}{{m}^{2aq-m}}$, as needed.
\end{proof}
Now $\ex[p(G,2q)] = \sum_{\{V_{i}\}}{\ex[f(G,\{V_{i}\})]}$. For every nonzero term $\ex[f(G,\{V_{i}\})]$, we have that $|\cup_{i} V_{i}| \leq 2aq-qy+z$. If $q > \frac{z}{y}$ then applying \lref{tboundlemma} with $m = 2aq - qy + z$, the number of non-zero terms $\ex[f(G,\{V_{i}\})]$ is at most 
\begin{align*}
\left({\frac{1}{{a}!}}\right)^{2q}\binom{2aq}{2aq-m}{n^{m}}{{m}^{2aq-m}} &\leq  \left({\frac{1}{{a}!}}\right)^{2q}{(2aqm)}^{2aq-m}{n^m} .\\
\end{align*}
Moreover, by our assumptions, each of these nonzero terms $E[f(G,\{V_i\})]$ has value at most $B^{2q}$, so $$\ex[p(G,2q)] \leq \left({\frac{1}{{a}!}}\right)^{2q}{(2aqm)}^{2aq-m}{n^m} B^{2q}.$$
Now, by Markov's inequality applied to $p(G,2q)$, 
$$\pr\sbkets{\sqrt[2q]{p(G,2q)} > \sqrt[2q]{\ex[p(G,2q)]/\epsilon}} < \epsilon.$$ 
We next choose a value $q$ so as to minimize our estimate on $\sqrt[2q]{\ex[p(G,2q)]/\epsilon}$. Specifically, we set $q = \lceil \ln(n^z/\epsilon)/2y \rceil$ (we arrive at this value by minimizing the general estimate as a function of $q$ by setting the derivative to $0$ - we spare the reader the details). As long as $n\geq 10$, this guarantees that $q > z/y$ so that
\begin{align*}
\sqrt[2q]{\ex[p(G,2q)]/\epsilon} &\leq \frac{B}{a!} \cdot \frac{1}{\epsilon^{1/2q}} \cdot (2aqm)^{a-m/2q} \cdot n^{m/2q}\\
&= \frac{B}{a!} \cdot n^{a-y/2} \cdot \bkets{\frac{n^z}{\epsilon}}^{1/2q} (2aqm)^{a-m/2q}\\
&\leq \frac{B}{a!} \cdot n^{a-y/2} \cdot \bkets{\frac{n^z}{\epsilon}}^{1/2q} \cdot (2aq)^y\\
&\leq \frac{B}{a!} \cdot n^{a-y/2} \cdot e^y \cdot (2a)^y \cdot \bkets{\frac{\ln(n^z/\epsilon)}{2y} + 1}^y.
\end{align*}
The claim now follows by rearranging the above bound.
\end{proof}
\subsection{Bounds on $||R_a||$}\label{rabounds}
In this subsection, we prove \clref{clm:maintechnicalL} using \tref{thm:productgraphtonorm}. For convenience, we restate \clref{clm:maintechnicalL} here with more precise constants.
\begin{theorem}\label{th:rabounds}
If $n \geq 100$, for all $\epsilon \in (0,1)$, 
$\pr\left[||R_a|| > 2^{a^2+2a+2}\ln{(\frac{n}{\epsilon})}n^{a - \frac{1}{2}}\right] < \epsilon$.
\end{theorem}

The core of the proof will be to bound $|\cup_{j=1}^{2q} V_{i_j}|$ for any term $\prod_{j=1}^{2q} R_a(V_{i_j}, V_{i_{j+1}})$ with non-zero expectation which appear in the expansion of $tr((R_a^T R_a)^q)$. We will do so by arguing that the constraint graph associated with the term has at most $2aq-q+1$ connected components, which we do by inductively decomposing $R_a$ as follows. 
\begin{definition}
Given a partition $(A,B)$ of $[1,n]$, define $R_{a,A,B}(V_1,V_2) = R_{a}(V_1,V_2)$ if $V_1 \subseteq A$ and $V_2 \subseteq B$ and $0$ otherwise.
\end{definition}
\begin{proposition}
$\sum_{A,B}{R_{a,A,B}} = 2^{n-2a}R_a$
\end{proposition}
\begin{proof}
$R_{a,A,B}(V_1,V_2) = R_a(V_1,V_2) = 0$ whenever $V_1$ and $V_2$ are not disjoint. For all disjoint $V_1$ and $V_2$, $R_{a,A,B}(V_1,V_2) = R_a(V_1,V_2)$ for $2^{n-2a}$ choices of $A$ and $B$ and is $0$ for the rest.
\end{proof}
\begin{corollary}\label{cor:parttowhole}
$||R_a|| \leq 2^{2a}\max_{A,B}{\{||R_{a,A,B}||\}}$
\end{corollary}
\begin{proof}
Since $R_a = 2^{2a-n}\sum_{A,B}{R_{a,A,B}}$, $||R_a|| \leq 2^{2a-n}\sum_{A,B}{||R_{a,A,B}||} \leq 2^{2a}\max_{A,B}{\{||R_{a,A,B}||\}}$
\end{proof}
Now given $A$ and $B$, take 
\begin{align*}
p(G,2q) = tr((R^T_{a,A,B}R_{a,A,B})^q) &= \sum_{\{V_{ij}: i \in [1,q], j \in [1,2]\}}{\prod_{i=1}^{q}{R^T_{a,A,B}(V_{i1},V_{i2})R_{a,A,B}(V_{i2},V_{(i+1)1})}} \\
&= \sum_{{\{V_{ij}: i \in [1,q], j \in [1,2]\}: \atop \forall i, V_{i1} \subseteq B} \atop {\forall i, V_{i2} \subseteq A}}
{\prod_{i=1}^{q}{R_{a}(V_{i1},V_{i2})R_{a}(V_{i2},V_{(i+1)1})}}
\end{align*} 
where we take $V_{(q+1)1} = V_{11}$. 

To simplify this expression, rename the sets of vertices as follows.
\begin{definition} \ 
\begin{enumerate}
\item If $i \in [1,2q]$ and $i$ is odd then take $W_i = V_{(\frac{i+1}{2})1}$
\item If $i \in [1,2q]$ and $i$ is even then take $W_i = V_{(\frac{i}{2})2}$
\end{enumerate}
\end{definition}
We now have that 
\begin{equation}\label{eq:termab}
p(G,2q) = \sum_{{{\{W_i}: i \in [1,2q]\}: \atop \forall \text{ odd } i, W_{i} \subseteq B} \atop {\forall \text { even } i, W_{i} \subseteq A}}{\prod_{i=1}^{2q}{R_{a}(W_{i},W_{i+1})}},
\end{equation}
where we take $W_{2q + 1} = W_1$. 
To study which of these terms may have non-zero expectation, we first define a graph related to the corresponding constraint graph. 
\begin{definition}
Given a constraint graph $C$, let $H$ be a graph with two types of edges, product edges and constraint edges, such that 
\begin{enumerate}
\item $V(H) = \{W_i: i \in [1,2q]\}$
\item $E_{P}(H) = \{(W_i,W_{i+1}): i \in [1,2q]\}$
\item $E_{C}(H) = \{(W_i,W_j): i \neq j, W_i \cap W_j \neq \emptyset\}$
\end{enumerate}
\end{definition}
Now, each $R_{a}(W_{i},W_{i+1})$ is a random variable with expectation $0$, so if any $R_{a}(W_{i},W_{i+1})$ is independent from everything else, the product will have expectation $0$. Such dependencies arise due to the presence of edges from G occurring in (at least two) different ``elements" (say $(W_i, W_{i+1})$, $(W_j, W_{j+1})$ for $i \neq j$) of the term. Such repeated occurrences manifest in our constraint graphs (and the graph $H$ defined above) as (three or four) cycles in the graph, which we call {\sl independence breaking}. For a term to have non-zero expectation it must be that every element $(W_i,W_{i+1})$ is on some such cycle. This implies that each product $\prod_{i=1}^{2q}{R_{a}(W_{i},W_{i+1})}$ has zero expected value unless all of the product edges in the corresponding $H$ are part of independence-breaking cycles. This places restrictions on $H$ (see \lref{lem:connectedcomponentbound}) which in turn places restrictions on the constraint graph $C$, allowing us to use \tref{thm:productgraphtonorm}. We make these ideas precise below.
\begin{definition}
Given $q$ and $\{W_1, \cdots, W_{2q}\}$, we define $W_{i \pm 2q} = W_{i}$ for all $i \in [1,2q]$.
\end{definition}
\begin{definition} If $q \geq 2$, 
\begin{enumerate}
\item Define an independence breaking 3-cycle in $H$ to consist of product edges $(W_i,W_{i+1})$, $(W_{i+1},W_{i+2})$ 
and a constraint edge $((W_{i},j),(W_{i+2},j))$.
\item Define an independence breaking 4-cycle to consist of product edges $e_1 = (W_{i_1},W_{i_1+1})$, 
$e_2 = (W_{i_2},W_{i_2 \pm 1})$ and constraint edges $(W_{i_1},W_{i_2})$ and $(W_{i_1 + 1},W_{i_2 \pm 1})$.
\end{enumerate}
\end{definition}
\begin{proposition}
For all $W_1, \cdots, W_{2q}$ such that $W_i \subseteq B$ whenever $i$ is odd and $W_i \subseteq A$ whenever $i$ is even, if the corresponding $H$ has a product edge $(W_i,W_{i+1})$ which is not contained in any independence-breaking cycle then $E[\prod_{i=1}^{2q}{R_{a}(W_{i},W_{i+1})}] = 0$
\end{proposition}
\begin{proof}
If $(W_i,W_{i+1})$ is not contained in any independence-breaking cycle then no edge between $W_i$ and $W_{i+1}$ appears anywhere else so $R_{a}(W_{i},W_{i+1})$ is a random variable with expectation $0$ which is independent from everything else and thus $E[\prod_{i=1}^{2q}{R_{a}(W_{i},W_{i+1})}] = 0$.
\end{proof}
We now bound the number of connected components in $H$ with the following lemma.
\begin{lemma}\label{lem:connectedcomponentbound}
Let $q \geq 2$ and $H$ be a graph such that
\begin{enumerate}
\item Every product edge of $H$ is contained in an independence-breaking cycle.
\item Every constraint edge of $H$ is of the form $(W_i,W_{i+j})$ where $j$ is even.
\end{enumerate} 
Then, the number of connected components in the graph defined by only the constraint edges of $H$ is at most $q+1$.
\end{lemma}
The intuitive idea behind this lemma is that if we add the constraint edges in the right order, every new constraint edge can put two product edges into independence breaking cycles. For example, a constraint edge between $W_{i-1}$ and $W_{i+1}$ puts the product edges $(W_{i-1},W_i)$ and $(W_i,W_{i+1})$ into an independence breaking 3-cycle. If we then add a constraint edge between $W_{i-2}$ and $W_{i+2}$, this puts the product edges $(W_{i-2},W_{i-1})$ and $(W_{i+1},W_{i+2})$ into an independence breaking 4-cycle. The final constraint edge can put 4 product edges into independence breaking cycles, so the number of constraint edges needed is $q-1$.

To make this argument work, we use an inductive proof. We note that if there is no $W_i$ which is isolated in $H$, we must have at least $q$ constraint edges. On the other hand, if there a $W_i$ which is isolated, there must be a constraint edge between $W_{i-1}$ and $W_{i+1}$. As noted above, this constraint edge puts the product edges $(W_{i-1},W_i)$ and $(W_i,W_{i+1})$ into an independence breaking 3-cycle. We take this to be the first constraint edge. We then argue that we can essentially delete $W_i$ and merge $W_{i-1}$ and $W_{i+1}$ which allows us to use the inductive hypothesis. We make these ideas rigorous below.
\begin{proof}[Proof of \lref{lem:connectedcomponentbound}]
We prove \lref{lem:connectedcomponentbound} by induction on $q$. The base case $q = 2$ is trivial, as we clearly need at least one constraint edge, so the number of connected components in $H$ is at most $3$. 
Now assume that $q = k \geq 3$ and the result is true for $q = k-1$.

First note that if there is no $W_i$ which is isolated (when looking only at constraint edges), then there are at most 
$q$ connected components in $H$. Thus, we may assume that $W_i$ is isolated for some $i$. Now 
note that for the product edge $(W_{i-1},W_{i})$, since $W_i$ is isolated, there are no independence breaking 3-cycles or 4-cycles where $W_i$ is the endpoint of a constraint edge. 
Thus, we must have that $(W_{i-1},W_{i})$ is part of an independence breaking 3-cycle consisting of $(W_{i-1},W_i)$, $(W_i,W_{i+1})$, and 
a constraint edge $(W_{i-1},W_{i+1})$.

Now form a new graph $H'$ as follows. Delete $W_i$ and contract the constraint edge between $W_{i-1}$ and $W_{i+1}$. More precisely, 
\begin{enumerate}
\item Take $V(H') = V(H) \setminus \{W_{i-1},W_i,W_{i+1}\} \cup \{U\}$
\item Take $E_{product}(H') = E_{product}(H) \setminus \{(W_{j},W_{j+1}): j \in [i-2,i+1]\} \cup \{(W_{i-2},U),(U,W_{i+2})\}$
\item Take \begin{align*}
E_{constraint}(H') &= E_{constraint}(H) \setminus \{(W_{i-1},W_{j}): (W_{i-1},W_{j}) \in E_{constraint}(H)\} \\
&\setminus \{(W_{i+1},W_{j}): (W_{i+1},W_{j}) \in E_{constraint}(H)\} \\
&\cup \{(U,W_{j}): (W_{i-1},W_{j}) \in E_{constraint}(H) \text{ or } (W_{i+1},W_{j}) \in E_{constraint}(H)\}
\end{align*}
\end{enumerate}
After doing this, rename $U$ as $W_{i-1}$ and rename each $W_{j}$ where $j > i+1$ as $W_{j-2}$. In going from $H$ to $H'$, we have effectively reduced both $q$ and the number of connected components by $1$. To complete the proof, we need to check that $H'$ satisfies the inductive hypotheses. Based on the reduction from $H$ to $H'$, we still have that every constraint edge is of the form $(W_i,W_{i+j})$ where $j$ is even. We check that every product edge is still part of an independence-breaking cycle case by case. 
\begin{enumerate}
\item Every independence-breaking cycle which did not contain the constraint edge $(W_{i-1},W_{i+1})$ in $H$ is preserved in $H'$ except that the vertices may have been renamed. The reason for this is that such an independence breaking cycle in $H$ cannot contain $W_i$ and can contain at most one of $\{W_{i-1},W_{i+1}\}$.
\item The independence-breaking 3-cycle in $H$ consisting of the product edges $(W_{i-1},W_{i})$, $(W_{i},W_{i+1})$ and the constraint edge $(W_{i-1},W_{i+1})$ is removed, but so are the product edges $(W_{i-1},W_{i})$ and $(W_{i},W_{i+1})$, so this is fine.
\item If we have an independence breaking 4-cycle in $H$ consisting of the product edges $(W_{i-2},W_{i-1})$, $(W_{i+1},W_{i+2})$ and the constraint edges $(W_{i-1},W_{i+1})$, $(W_{i-2},W_{i+2})$, this becomes an independence-breaking 3-cycle in $H'$ with product edges $(W_{i-2},W_{i-1})$, $(W_{i-1},W_{i})$ and a constraint edge $(W_{i-2},W_{i})$ (note that $W_{i-1}$ and $W_{i+1}$ are merged into $W_{i-1}$ in $H'$ and $W_{i+2}$ is renamed as $W_i$ in $H'$).
\end{enumerate}
$H'$ satisfies the inductive hypotheses, so looking only at the constraint edges, $H'$ has at most $(q-1)+1 = q$ connected components. $H$ has one more connected component than $H'$ (the vertex $W_i$ in $H$), so $H$ has at most $q+1$ connected components, as needed.
\end{proof}
The above lemma combined with \lref{lm:componentstounion} gives the following corollary.
\begin{corollary}
For all terms $\prod_{i=1}^{2q}R_a(W_i,W_{i+1})$ occurring in \eref{eq:termab} with nonzero expectation, $|\cup_{i=1}^{2q} W_i| \leq 2aq - (2q) + q + 1$. 
\end{corollary}
\ignore{We now use this bound on the number of connected components of $H$ to prove a corresponding bound on the number of connected components of $C$.
\begin{proposition}
For all constraint graphs $C$ and the corresponding graphs $H$, $|V(C)|$ minus the number of connected components of $C$ is at least as large as $|V(H)|$ minus the number of connected components of $H$.
\end{proposition}
\begin{proof}
Consider the process of starting with $V(H)$ and then adding edges of $H$ one at a time where each new edge reduces the number of connected components until this is no longer possible. The number of edges we add is $|V(H)|$ minus the number of connected components of $H$. Now note that if we start with $V(C)$ and add the corresponding edges in $C$, this already reduces the number of connected components by the number of edges we added and the result follows.
\end{proof}
\begin{corollary}
For all terms in \eref{eq:termab} with nonzero expectation, the corresponding constraint graph $C$ has at most $2aq - q + 1$ connected components 
\end{corollary}}
We can now prove \tref{th:rabounds}
\begin{proof}[Proof of \tref{th:rabounds}]
We can now apply \tref{thm:productgraphtonorm} with $y = 1$, and $z = 1$ by the above corollary. Every entry of $R_{a,A,B}$ has magnitude at most $2^{a^2}$ so we can take $B = 2^{a^2}$. By \tref{thm:productgraphtonorm}, if $n \geq 10$, for all $A$ and $B$, for every $\epsilon \in (0,1)$, 
$$\pr\left[||R_{a,A,B}|| > \frac{2^{a^2}}{a!}\left(2ea\Big(\frac{\ln{n} - \ln{\epsilon}}{2} + 1\Big)\right)
n^{a-1/2}\right] = 
\pr\left[||R_{a,A,B}|| > \frac{2^{a^2}ea}{a!}(\ln{n} - \ln{\epsilon} + 2)
n^{a - \frac{1}{2}}\right] < \epsilon$$
Since $4\ln{n} \geq e(\ln{n} + 2)$ for all $n \geq 100$ and $a! \geq a$, we have that for all $n \geq 100$, for all $A$ and $B$ and all $\epsilon \in (0,1)$, 
$$\pr\left[||R_{a,A,B}|| > 2^{a^2+2}\ln{(\frac{n}{\epsilon})}n^{a - \frac{1}{2}}\right] < \epsilon$$
Now by \cref{cor:parttowhole}, $||R_a|| \leq 2^{2a}\max_{A,B}{\{||R_{a,A,B}||\}}$ so $$\pr\left[||R_{a}|| > 2^{a^2 + 2a + 2}\ln{(\frac{n}{\epsilon})}n^{a - \frac{1}{2}}\right] < \epsilon$$
\end{proof}




\section{Concentration bounds for number of cliques and $deg_G(I)$}\label{sec:degreevariancecalculations}
We now prove large deviation bounds for $deg_G(\;)$ leading to \clref{clm:degreebound} which we state below in a more precise form.

\begin{theorem}\label{th:degreeconcentrationfinal}
If $n \geq 10$, and $\epsilon \in (0,1)$, then for all $I \subseteq [n]$, with $|I| = i \leq 2r$, 
$$\pr\sbkets{\left|deg_G(I) - 2^{-\binom{2r}{2} + \binom{i}{2}} \cdot \binom{n-i}{2r-i}\right| > 2(\ln(128/\epsilon))^2 n^{2r-i-1/2} \,\left|\right.\, \text{ ($I$ is a clique)}} \leq \epsilon.$$
\end{theorem}

To prove the claim we first show a similar concentration bound for the number of cliques of a certain size in $G$.  While similar results appear in the literature, see for instance \cite{Rucinski88, Vu01, JLR11book}, we give a short direct proof based on \tref{thm:productgraphtonorm}. 
\begin{definition}
For a graph $G$, define $N_a(G)$ to be the number of $a$-cliques in $G$.
\end{definition}
\begin{theorem}\label{th:cliquelargedeviation}
For all $a$, for all $n \geq 10$ and $\epsilon \in (0,1)$, $E[N_a(G)] = 2^{-\binom{a}{2}}{\binom{n}{a}}$ and 
$$\pr\left[|N_a(G) - E[N_a(G)]| > (\ln(64/\epsilon))^2 \cdot n^{a-1}\right] < \epsilon.$$
\end{theorem}
\begin{proof}
The first part of the theorem is trivial so we focus on the second part. Given a set of vertices $V$ of size $a$, define $c_V$ to be $1-2^{-\binom{a}{2}}$ if $V$ is a clique and $-2^{-\binom{a}{2}}$ otherwise. Then, 
$$N_a(G) - E[N_a(G)] = \sum_{V:|V| = a}{c_V}.$$
Now let's consider the function 
$p(G,2q) = (\sum_{V:|V| = a}{c_V})^{2q} =\sum_{W_1, \cdots, W_{2q}}{\prod_{i=1}^{2q}{c_{W_i}}}$.

Note that $E[\prod_{i=1}^{2q}{c_{W_i}}] = 0$ unless each set of vertices $W_i$ has two vertices in common with a different set of vertices $W_j$. Now consider a graph $C_2$ where the vertices are $\{W_1,\ldots,W_{2q}\}$ and an edge between $W_i, W_j$ if $|W_i \cap W_j| \geq 2$. Let $t$ be the number of connected components in $C_2$. We claim that $|\cup_i W_i| \leq 2aq - 4q + 2t$. For, as in the proof of \lref{lm:componentstounion}, first consider elements $W_{i_1},\ldots,W_{i_t}$ belonging to the $t$ different connected components. Now, add the remaining elements of $\{W_1,\ldots,W_{2q}\}$ so that each new element is adjacent to at least one of the previously added sets. When doing so, each step can increase the size of the union by at most $a-2$. Therefore, the size of the union is at most $a t + (a-2)(2q-t) = 2aq - 4q + 2t$. On the other hand, each connected component in $C_2$ must have at least two vertices, so $t \leq q$.  Therefore, $|\cup_i W_i| \leq 2aq - 2q$. 

We can now apply \tref{thm:productgraphtonorm} with $y=2$, $z=0$ and $B=1$ so that for $n \geq 10$, and $\epsilon \in (0,1)$, 
$$\pr\left[|N_a(G) - E[N_a(G)]| > \frac{1}{a!} \cdot \left(ea\Big(\frac{-\ln{\epsilon}}{4} + 1\Big)\right)^2 \cdot n^{a-1}\right] < \epsilon.$$
Using the facts that $e^2 < 8$ and $\frac{m^2}{m!} \leq 2$ for all nonnegative integers $m$, we have that 
$$\pr\left[|N_a(G) - E[N_a(G)]| > (\ln(64/\epsilon))^2 \cdot n^{a-1}\right] < \epsilon.$$
\end{proof}

We are now ready to prove \tref{th:degreeconcentrationfinal}. The idea is as follows. Let $A_I$ be the collection of vertices which are adjacent to all the vertices in $I$. Then, conditioned on $I$ being a clique, $deg_G(I)$ is just the number of cliques of size $2r-i$ in the vertices $A_I$ which is primarily determined by $|A_I|$. This is because the edges between vertices of $A_I$ are independent of the edges involving vertices in $I$ so that we can apply \tref{th:cliquelargedeviation} to $A_I$. 

\begin{proof}[Proof of \tref{th:degreeconcentrationfinal}]
Let $A_I$ be as above and let us condition on $I$ being a clique. Then, $deg_G(I)$ is just the number of cliques of size $2r-i$ among the vertices in $A_I$. Therefore, by \tref{th:cliquelargedeviation}, with probability at least $1-\epsilon/2$, 
$$\left| deg_G(I) - 2^{-\binom{2r-i}{2}} \binom{|A_I|}{2r-i} \right| \leq (\ln(128/\epsilon))^2 \cdot n^{2r-i-1}.$$

We next argue that $\binom{|A_I|}{2r-i}$ is concentrated around its mean. For $j \notin I$, let $X_j$ be the indicator random variable that is $1$ if the $j$'th vertex is adjacent to all the vertices in $I$ and $0$ otherwise. Then, $|A_I| = \sum_{j \notin I} X_j$ and
$$\binom{|A_I|}{2r-i} = \sum_{J \subseteq [n] \setminus I, |J| = 2r-i} \prod_{j \in J} X_j \equiv f(\{X_j: j \notin I\}).$$
Observe that the random variables $X_j$ are independent of each other and that 
$$\ex[f(\{X_j:J \notin I\})] = 2^{-i(2r-i)} \binom{n-i}{2r-i}.$$
We next apply McDiarmid's inequality to the function $f$. Note that changing any single coordinate of the inputs to $f$ can change its value by at most $n^{2r-i-1}$. Therefore, by \tref{th:mcdiarmid}, with probability at least $1-\epsilon/2$, 
$$\left|\binom{|A_I|}{2r-i} - 2^{-i(2r-i)} \binom{n-i}{2r-i} \right| \leq \sqrt{\ln(4/\epsilon)} \cdot n^{2r-i-.5}.$$
Combining the above equations, we get that with probability at least $1-\epsilon$,
\begin{multline*}
\left|deg_G(I) - 2^{-\binom{2r-i}{2} - i(2r-i)} \binom{n-i}{2r-i}\right| \leq (\ln(128/\epsilon))^2 \cdot n^{2r-i-1} + 2^{-\binom{2r-i}{2}} \cdot \sqrt{\ln(4/\epsilon)} \cdot n^{2r-i-.5} \leq \\
2 \ln((128/\epsilon)^2) \cdot n^{2r-i-.5}.
\end{multline*}
The theorem now follows as $\binom{2r-i}{2} + i(2r-i) = \binom{2r}{2} - \binom{i}{2}$. 
\end{proof}

\section{Conclusion and future work}
In this work we showed a lower bound for the maximum clique problem on random $G(n,1/2)$ graphs in the $\sos$ hierarchy and positivstellensatz proof system. Besides the specific application to clique lower bounds, the PSD'ness of the matrix $\mm$ from \eref{eq:mainmatrixdef} seems to carry further information that could be potentially useful elsewhere, perhaps for studying various sub-graph statistics. Further, the arguments related to association schemes and bounding the norm of locally random matrices could also be useful elsewhere, especially for other $\sos$ hierarchy lower bounds. One natural and interesting candidate is the densest subgraph problem.

For planted clique itself, the most obvious open problem is to tighten the gap between the current upper bound of $O(\sqrt{n}/2^r)$ and our lower bound of $2^{-O(r)}(\sqrt{n}/\log n)^{1/r}$ for $r$ rounds of the SOS hierarchy. In particular, can a constant number of rounds of $\sos$ beat the square-root barrier and identify planted cliques of size $o(n^{1/2})$? 
Kelner\footnote{Personal comminication} showed that our dual certificate $\mm$ actually is not PSD for $k$ roughly $O(n^{1/(r+1)})$. Thus one needs to come up with a different dual certificate to approach the upper bound of $\sqrt{n}$ even for $r=2$. 

\paragraph{Acknowledgements} We thank Boaz Barak, Siu-on Chan, Jonathan Kelner, Robert Krauthgamer, James Lee, Nati Linial, David Steurer, Madhu Sudan and Amir Yehudayoff for several useful comments.

\bibliographystyle{alpha}
\bibliography{proofcomplexity}

\newcommand{\etalchar}[1]{$^{#1}$}
\begin{thebibliography}{AAK{\etalchar{+}}07}

\bibitem[AAK{\etalchar{+}}07]{AlonAKMRX07}
Noga Alon, Alexandr Andoni, Tali Kaufman, Kevin Matulef, Ronitt Rubinfeld, and
  Ning Xie.
\newblock Testing k-wise and almost k-wise independence.
\newblock In {\em Proceedings of the 39th Annual {ACM} Symposium on Theory of
  Computing, San Diego, California, USA, June 11-13, 2007}, pages 496--505,
  2007.

\bibitem[ABBG10]{AroraBBG10}
Sanjeev Arora, Boaz Barak, Markus Brunnermeier, and Rong Ge.
\newblock Computational complexity and information asymmetry in financial
  products.
\newblock In {\em ICS}, pages 49--65, 2010.

\bibitem[ABS10]{AroraBS10}
Sanjeev Arora, Boaz Barak, and David Steurer.
\newblock Subexponential algorithms for unique games and related problems.
\newblock In {\em FOCS}, pages 563--572, 2010.

\bibitem[ABW10]{AppleBW10}
Benny Applebaum, Boaz Barak, and Avi Wigderson.
\newblock Public-key cryptography from different assumptions.
\newblock In {\em Proceedings of the forty-second ACM symposium on Theory of
  computing}, pages 171--180. ACM, 2010.

\bibitem[AKS98]{AlonKS98}
Noga Alon, Michael Krivelevich, and Benny Sudakov.
\newblock Finding a large hidden clique in a random graph.
\newblock {\em Random Struct. Algorithms}, 13(3-4):457--466, 1998.

\bibitem[Art27]{Artin27}
Emil Artin.
\newblock Uber die zerlegung deﬁniter funktionen in quadrate.
\newblock {\em Abhandlungen aus dem Mathematischen Seminar der Universitat
  Hamburg}, 5(1):100--115, 1927.

\bibitem[BBH{\etalchar{+}}12]{BarakBHKSZ12}
Boaz Barak, Fernando G. S.~L. Brand{\~a}o, Aram~Wettroth Harrow, Jonathan~A.
  Kelner, David Steurer, and Yuan Zhou.
\newblock Hypercontractivity, sum-of-squares proofs, and their applications.
\newblock In {\em STOC}, pages 307--326, 2012.

\bibitem[BCR98]{BCR98}
Jacek Bochnak, Michel Coste, and Marie-Francoise Roy.
\newblock {\em Real algebraic geometry}.
\newblock Springer, 1998.

\bibitem[BCV{\etalchar{+}}12]{BhaskaraCVGZ12}
Aditya Bhaskara, Moses Charikar, Aravindan Vijayaraghavan, Venkatesan
  Guruswami, and Yuan Zhou.
\newblock Polynomial integrality gaps for strong {SDP} relaxations of densest
  {\it k}-subgraph.
\newblock In {\em SODA}, pages 388--405, 2012.

\bibitem[BR13]{BerthetR13}
Q.~Berthet and P.~Rigollet.
\newblock Complexity theoretic lower bounds for sparse principal component
  detection.
\newblock {\em J. Mach. Learn. Res., W and CP}, 30:1046--1066 (electronic),
  2013.

\bibitem[BRS11]{BarakRS11}
Boaz Barak, Prasad Raghavendra, and David Steurer.
\newblock Rounding semidefinite programming hierarchies via global correlation.
\newblock In {\em FOCS}, pages 472--481, 2011.

\bibitem[DGGP14]{DGY14}
YAEL DEKEL, ORI GUREL-GUREVICH, and YUVAL PERES.
\newblock Finding hidden cliques in linear time with high probability.
\newblock {\em Combinatorics, Probability and Computing}, 23:29--49, 1 2014.

\bibitem[DM15]{DeshpandeM15}
Yash Deshpande and Andrea Montanari.
\newblock Improved sum-of-squares lower bounds for hidden clique and hidden
  submatrix problems.
\newblock {\em CoRR}, abs/1502.06590, 2015.

\bibitem[FGR{\etalchar{+}}13]{FeldmanGRVX13}
Vitaly Feldman, Elena Grigorescu, Lev Reyzin, Santosh Vempala, and Ying Xiao.
\newblock Statistical algorithms and a lower bound for detecting planted
  cliques.
\newblock In {\em STOC}, pages 655--664, 2013.

\bibitem[FK00]{FeigeK00}
Uriel Feige and Robert Krauthgamer.
\newblock Finding and certifying a large hidden clique in a semirandom graph.
\newblock {\em Random Struct. Algorithms}, 16(2):195--208, 2000.

\bibitem[FK03]{FeigeK03}
Uriel Feige and Robert Krauthgamer.
\newblock The probable value of the {L}ov{\'a}sz--{S}chrijver relaxations for
  maximum independent set.
\newblock {\em SIAM J. Comput.}, 32(2):345--370, 2003.

\bibitem[FK08]{FK08}
Alan~M. Frieze and Ravi Kannan.
\newblock A new approach to the planted clique problem.
\newblock In {\em {IARCS} Annual Conference on Foundations of Software
  Technology and Theoretical Computer Science, {FSTTCS} 2008, December 9-11,
  2008, Bangalore, India}, pages 187--198, 2008.

\bibitem[God]{GodsilNotes2}
Chris Godsil.
\newblock Association schemes.
\newblock Lecture Notes available at
  http://quoll.uwaterloo.ca/mine/Notes/assoc1.pdf.

\bibitem[Gri01a]{Grigoriev01}
Dima Grigoriev.
\newblock Complexity of positivstellensatz proofs for the knapsack.
\newblock {\em Computational Complexity}, 10(2):139--154, 2001.

\bibitem[Gri01b]{Grigoriev01b}
Dima Grigoriev.
\newblock Linear lower bound on degrees of positivstellensatz calculus proofs
  for the parity.
\newblock {\em Theor. Comput. Sci.}, 259(1-2):613--622, 2001.

\bibitem[GS11]{GuruswamiS11}
Venkatesan Guruswami and Ali~Kemal Sinop.
\newblock Lasserre hierarchy, higher eigenvalues, and approximation schemes for
  graph partitioning and quadratic integer programming with {PSD} objectives.
\newblock In {\em FOCS}, pages 482--491, 2011.

\bibitem[GV01]{GrigorievV01}
Dima Grigoriev and Nicolai Vorobjov.
\newblock Complexity of null-and positivstellensatz proofs.
\newblock {\em Ann. Pure Appl. Logic}, 113(1-3):153--160, 2001.

\bibitem[GVL96]{Golub1996}
G.H. Golub and C.F. Van~Loan.
\newblock {\em Matrix Computations}.
\newblock Johns Hopkins Studies in the Mathematical Sciences. Johns Hopkins
  University Press, 1996.

\bibitem[Jer92]{Jerrum92}
Mark Jerrum.
\newblock Large cliques elude the metropolis process.
\newblock {\em Random Struct. Algorithms}, 3(4):347--360, 1992.

\bibitem[JLR11]{JLR11book}
S.~Janson, T.~Luczak, and A.~Rucinski.
\newblock {\em Random Graphs}.
\newblock Wiley Series in Discrete Mathematics and Optimization. Wiley, 2011.

\bibitem[Kar76]{Karp76}
R.~M. Karp.
\newblock Probabilistic analysis of some combinatorial search problems.
\newblock {\em In: Algorithms and Complexity: New Directions and Recent
  Results}, pages 1--19, 1976.

\bibitem[Kri64]{Krivine64}
Jean-Louis Krivine.
\newblock Anneaux preordonn´es.
\newblock {\em Journal d'Analyse Mathematique}, 12(1):307--326, 1964.

\bibitem[Kuc95]{Kucera95}
Ludek Kucera.
\newblock Expected complexity of graph partitioning problems.
\newblock {\em Discrete Applied Mathematics}, 57(2-3):193--212, 1995.

\bibitem[Las01]{Lasserre01}
Jean~B. Lasserre.
\newblock Global optimization with polynomials and the problem of moments.
\newblock {\em SIAM Journal on Optimization}, 11(3):796--817, 2001.

\bibitem[LS91]{LovaszS91}
L.~Lov\'asz and A.~Schrijver.
\newblock Cones of matrices and set-functions and $0$-$1$ optimization.
\newblock {\em SIAM Journal on Optimization}, 1(2):166--190, 1991.

\bibitem[MW13]{MekaW13}
Raghu Meka and Avi Wigderson.
\newblock Association schemes, non-commutative polynomial concentration, and
  sum-of-squares lower bounds for planted clique.
\newblock {\em CoRR}, abs/1307.7615, 2013.

\bibitem[OZ13]{ODonnellZ13}
Ryan O'Donnell and Yuan Zhou.
\newblock Approximability and proof complexity.
\newblock In {\em SODA}, pages 1537--1556, 2013.

\bibitem[Par00]{Parrilo00}
Pablo Parrilo.
\newblock Structured semidefinite programs and semialgebraic geometry methods
  in robustness and optimization, 2000.
\newblock PhD thesis, California Institute of Technology.

\bibitem[PS{\etalchar{+}}00]{Pevzner00}
Pavel~A Pevzner, Sing-Hoi Sze, et~al.
\newblock Combinatorial approaches to finding subtle signals in dna sequences.
\newblock In {\em ISMB}, volume~8, pages 269--278, 2000.

\bibitem[Put93]{Putinar93}
Mihai Putinar.
\newblock Positive polynomials on compact semi-algebraic sets.
\newblock {\em Indiana University Mathematics Journal}, 42(3):969--984, 1993.

\bibitem[Ruc88]{Rucinski88}
Andrzej Ruciński.
\newblock When are small subgraphs of a random graph normally distributed?
\newblock {\em Probability Theory and Related Fields}, 78(1):1--10, 1988.

\bibitem[SA90]{SheraliA90}
H.~Sherali and W.~Adams.
\newblock A hierarchy of relaxations between the continuous and convex hull
  representations for zero-one programming problems.
\newblock {\em SIAM Journal on Discrete Mathematics}, 3(3):411--430, 1990.

\bibitem[Sch91]{Schmudgen91}
Konrad Schmudgen.
\newblock The k-moment problem for compact semi-algebraic sets.
\newblock {\em Mathematische Annalen}, 289(1):203--206, 1991.

\bibitem[Sch08]{Schoenebeck08}
Grant Schoenebeck.
\newblock Linear level {L}asserre lower bounds for certain k-{CSP}s.
\newblock In {\em FOCS}, pages 593--602, 2008.

\bibitem[Ste73]{Stengle73}
Gilbert Stengle.
\newblock A nullstellensatz and a positivstellensatz in semialgebraic geometry.
\newblock {\em Mathematische Annalen}, 207(2):87--97, 1973.

\bibitem[Tul09]{Tulsiani09}
Madhur Tulsiani.
\newblock {CSP} gaps and reductions in the {L}asserre hierarchy.
\newblock In {\em STOC}, pages 303--312, 2009.

\bibitem[Ver]{Vershynin06}
Roman Vershynin.
\newblock Lecture 6: Norm of a random matrix.
\newblock Lecture Notes on Non-Asymptotic Random Matrix Theory - Available
  online.

\bibitem[vLW01]{VanLintW01}
J.H. van Lint and R.M. Wilson.
\newblock {\em A Course in Combinatorics}.
\newblock Cambridge University Press, 2001.

\bibitem[Vu01]{Vu01}
Van Vu.
\newblock A large deviation result on the number of small subgraphs of a random
  graph.
\newblock {\em Combinatorics, Probability and Computing}, 10:79--94, 1 2001.

\end{thebibliography}
\section{Hierarchy Gaps and Positivstellensatz Refutations}\label{sec:psdtogaps}
For a detailed discussion of the hierarchies and $\psd$-refutations we refer the reader to the discussions in \cite{ODonnellZ13}. The basic principle is that, typically, $\psd$-refutations are more robust and stronger than the hierarchy formulations. 

The $\sos$ (or Lasserre) relaxation for maximum clique is stated in \fref{fig:lasserresdp} (cf.~\cite{Tulsiani09}). Although, the formulation itself is not in terms of an SDP, it is a standard fact that as the program only involves inner products of vectors, the optimization can be done by semi-definite programming. 
\begin{figure}\label{fig:lasserresdp}
  \begin{mybox}
$\sos$-relaxation for Max-Clique. Input: Graph $G = (V,E)$, $r$ - number of rounds. Variables of the SDP are vectors $\bm{U}_S$, where $S \subseteq [n], |S| \leq r$.
$$\text{maximize } \sum_{i \in V} \|\bm{U}_{\{i\}}\|_2^2.$$
\begin{align*}
  \text{such that }\qquad \iprod{\bm{U}_{\{i\}}}{\bm{U}_{\{j\}}} &= 0, &&\forall i,j,\;\; \{i,j\} \notin E\\
\iprod{\bm{U}_{S_1}}{\bm{U}_{S_2}} &= \iprod{\bm{U}_{S_3}}{\bm{U}_{S_4}},  &&S_1 \cup S_2 = S_3 \cup S_4, |S_1 \cup S_2| \leq r\\
\iprod{\bm{U}_{S_1}}{\bm{U}_{S_2}} &\in [0,1], &&|S_1|, |S_2| \leq r\\
 \|\bm{U}_\emptyset\|_2^2 &= 1&&
\end{align*}
  \end{mybox}
\caption{$r$-round $\sos$-relaxation for Maximum Clique}
\end{figure}
The connection between \fref{fig:lasserresdp} and $\psd$-refutations comes from the following straightforward lemma stating that a certificate for $\psd$-refutations is simply a primal solution to the standard $r$-round $\sos$-relaxation of the problem.
\begin{lemma}
Let $G = (V,E)$ be a graph and let $\clq(G,k)$ denote the clique axioms as defined by Equations \ref{eq:maxclique}. Suppose that there exists a dual certificate $\m:\pnr \rgta \reals$ for $\clq(G,k)$ as defined in \dref{def:dualcert}. Then, the value of the $r$-round $\sos$-relaxation for maximum clique given by \fref{fig:lasserresdp} is at least $k$.
\end{lemma}
\begin{proof}
Let $\m:\pnr \rgta \reals$ be the dual certificate and $\mm \in \rnlr$ be the corresponding PSD matrix. Without loss of generality suppose that $\mm(\emptyset,\emptyset) = 1$. Let $\mm = U U^\dagger$, where $U = \reals^{\bnlr \times N}$ for some $N$. Finally, for $S \in \bnlr$, let $\bm{U}_S$ be the $S$'th row of $U$. We claim that the collection $(\bm{U}_S$, $|S| \leq r)$ gives a feasible solution for the SDP in \fref{fig:lasserresdp}. 

Observe that for any two subsets $S_1, S_2 \in \bnlr$, 
$$\iprod{\bm{U}_{S_1}}{\bm{U}_{S_2}} = \mm(S_1, S_2) = \m(X_{S_1 \cup S_2}).$$
Therefore, the vectors $(\bm{U}_S: |S| \leq r)$ satisfy the first two constraints of \fref{fig:lasserresdp} as $\m$ is a dual certificate. Further, $\|\bm{U}_\emptyset\|^2 = \mm(\emptyset,\emptyset) = 1$ and for any set $S$,
$$\|\bm{U}_S\|_2^2 = \iprod{\bm{U}_S}{\bm{U}_S} = \iprod{\bm{U}_S}{\bm{U}_\emptyset} \leq \|\bm{U}_S\|_2,$$
so that $\|\bm{U}_S\| \leq 1$. Thus, $(\bm{U}_S$: $|S| \leq r)$ give a feasible solution for the program in \fref{fig:lasserresdp}. Finally, the value of the solution is 
$$\sum_{i \in V} \|\bm{U}_{\{i\}}\|_2^2 = \sum_{i \in V} \m(X_{\{i\}}) = k.$$
This proves the lemma.
\end{proof}

Our main theorems now follow.
\begin{proof}[Proof of \tref{th:mainhierarchy}]
Let $G \lfta G(n,1/2)$. Then, from the above lemma and the proof of \tref{th:main} (where we showed the existence of a dual certificate for the clique axioms), the value of the $r$-round $\sos$-relaxation for max-clique on $G$ is at least $n^{1/2r}/C^{r} (\log n)^{1/r}$ with high probability. The claim follows as the integral value is $(2+o(1))\log_2 n$ with high probability. 
\end{proof}

\begin{proof}[Proof of \cref{cor:mainhierarchy}]
The value of the relaxation in \fref{fig:lasserresdp} is clearly monotone with respect to adding edges. Therefore, from the above argument, for $G \lfta G(n,1/2,t)$ the value of the $r$-round $\sos$-relaxation for max-clique on $G$ is at least $n^{1/2r}/C^{r} (\log n)^{1/r}$ with high probability. The claim follows as the integral value is $t$ with high probability.
\end{proof}

\end{document}